\documentclass[reqno,11pt]{amsart}
\usepackage{amsmath, latexsym, amsfonts, amssymb, amsthm, amscd}
\usepackage{mathrsfs}
\usepackage{graphics,epsf,psfrag}
\usepackage[lofdepth,lotdepth]{subfig}
\usepackage{graphicx}

\numberwithin{equation}{section}

\newtheorem{theorem}{Theorem}[section]
\newtheorem{lemma}[theorem]{Lemma}
\newtheorem{proposition}[theorem]{Proposition}

\newtheorem{rem}[theorem]{Remark}

\newcommand{\ens}[2]{\left\{#1\,;\,#2\right\}}

\newcommand{\R}{\mathbb{R}}

\renewcommand{\tilde}{\widetilde}

\newcommand{\cA}{{\ensuremath{\mathcal A}} }
\newcommand{\cB}{{\ensuremath{\mathcal B}} }
\newcommand{\cF}{{\ensuremath{\mathcal F}} }

\newcommand{\cE}{{\ensuremath{\mathcal E}} }

\newcommand{\cC}{{\ensuremath{\mathcal C}} }
\newcommand{\cN}{{\ensuremath{\mathcal N}} }
\newcommand{\cM}{{\ensuremath{\mathcal M}} }

\newcommand{\cD}{{\ensuremath{\mathcal D}} }
\newcommand{\cU}{{\ensuremath{\mathcal U}} }
\newcommand{\cV}{{\ensuremath{\mathcal V}} }
\newcommand{\cR}{{\ensuremath{\mathcal R}} }

\newcommand{\bC}{{\ensuremath{\mathbf C}} }

\DeclareMathSymbol{\leqslant}{\mathalpha}{AMSa}{"36} 
\DeclareMathSymbol{\geqslant}{\mathalpha}{AMSa}{"3E} 
\DeclareMathSymbol{\eset}{\mathalpha}{AMSb}{"3F}     
\renewcommand{\leq}{\;\leqslant\;}                   
\renewcommand{\geq}{\;\geqslant\;}                   
\newcommand{\dd}{\,\text{\rm d}}             
\DeclareMathOperator*{\Supp}{Supp}
\newcommand{\suptwo}[2]{\sup_{\substack{#1 \\ #2}}} 

\newcommand{\bbC}{{\ensuremath{\mathbb C}} }

\newcommand{\bbR}{{\ensuremath{\mathbb R}} }
\newcommand{\bbS}{{\ensuremath{\mathbb S}} }

\newcommand{\ga}{\alpha}
\newcommand{\gb}{\beta}
\newcommand{\gd}{\delta}
\newcommand{\gep}{\varepsilon}       
\newcommand{\gp}{\varphi}

\newcommand{\gD}{\Delta}

\newcommand{\go}{\omega}

\newcommand{\gl}{\lambda}

\newcommand{\gs}{\sigma}

\newcommand{\gt}{\vartheta}
\newcommand{\gtta}{\theta}

\makeatletter
\def\captionfont@{\footnotesize}
\def\captionheadfont@{\scshape}

\long\def\@makecaption#1#2{%
  \vspace{2mm}
  \setbox\@tempboxa\vbox{\color@setgroup
    \advance\hsize-6pc\noindent
    \captionfont@\captionheadfont@#1\@xp\@ifnotempty\@xp
        {\@cdr#2\@nil}{.\captionfont@\upshape\enspace#2}%
    \unskip\kern-6pc\par
    \global\setbox\@ne\lastbox\color@endgroup}%
  \ifhbox\@ne 
    \setbox\@ne\hbox{\unhbox\@ne\unskip\unskip\unpenalty\unkern}%
  \fi
  \ifdim\wd\@tempboxa=\z@ 
    \setbox\@ne\hbox to\columnwidth{\hss\kern-6pc\box\@ne\hss}%
  \else 
    \setbox\@ne\vbox{\unvbox\@tempboxa\parskip\z@skip
        \noindent\unhbox\@ne\advance\hsize-6pc\par}%
\fi
  \ifnum\@tempcnta<64 
    \addvspace\abovecaptionskip
    \moveright 3pc\box\@ne
  \else 
    \moveright 3pc\box\@ne
    \nobreak
    \vskip\belowcaptionskip
  \fi
\relax
}
\makeatother
\def\writefig#1 #2 #3 {\rlap{\kern #1 truecm
\raise #2 truecm \hbox{#3}}}

\newcommand{\gap}{\gl_{K}}

\newcommand{\cro}[2]{\ensuremath{\left\langle#1\, ,\, #2\right\rangle}} 

\newcommand{\crosqmu}[2]{\ensuremath{\left\langle#1\, ,\, #2\right\rangle_{ -1, {q_0},\mu}}}
\newcommand{\crosq}[2]{\ensuremath{\left\langle#1\, ,\, #2\right\rangle_{ -1, {q_0}}}}

\newcommand{\croLqmu}[2]{\ensuremath{\left\langle#1\, ,\, #2\right\rangle_{ 2, {q_0},\mu}}} 
\newcommand{\croLq}[2]{\ensuremath{\left\langle#1\, ,\, #2\right\rangle_{ 2, {q_0}}}}

\newcommand{\Nsqmu}[1]{\ensuremath{\left\|\,#1\,\right\|_{ -1, {q_0}, \mu}}}
\newcommand{\NLqmu}[1]{\left\|\,#1\,\right\|_{2, {q_0}, \mu}}
\newcommand{\NLmu}[1]{\left\|\,#1\,\right\|_{2, \mu}}
\newcommand{\Ninf}[1]{\left\|\,#1\,\right\|_{\infty}}

\newcommand{\intS}{\int_{\bbS}}
\newcommand{\intR}{\int_{\bbR}}
\newcommand{\intSR}{\int_{\bbR\times{\bbS }}}

\newcommand{\In}{\ensuremath{I_n}}
\newcommand{\Jn}{\ensuremath{J_n}}

\begin{document}

\title[Coupled noisy oscillators 
and random natural frequencies]{Coherence stability  
and effect of random natural frequencies in populations of coupled oscillators}

\author{Giambattista Giacomin}
\address{
  Universit{\'e} Paris Diderot (Paris 7) and Laboratoire de Probabilit{\'e}s et Mod\`eles Al\'eatoires (CNRS),
U.F.R.                Math\'ematiques, Case 7012 (site Chevaleret)
             75205 Paris Cedex 13, France
}
\author{Eric Lu\c{c}on}
\address{Universit{\'e} Paris 6 -- Pierre et Marie Curie and Laboratoire de Probabilit{\'e}s et Mod\`eles Al\'eatoires (CNRS U.M.R. 7599), U.F.R.
Mathematiques, Case 188, 4 place
Jussieu, 75252 Paris cedex 05, France
}

\author{Christophe Poquet}
\address{
  Universit{\'e} Paris Diderot (Paris 7) and Laboratoire de Probabilit{\'e}s et Mod\`eles Al\'eatoires (CNRS),
U.F.R.                Math\'ematiques, Case 7012 (site Chevaleret)
             75205 Paris Cedex 13, France
}

\date{\today}

\begin{abstract}
We consider the (noisy) Kuramoto model, that is   a population of $N$ oscillators, or rotators,
 with mean-field interaction. Each oscillator has its own randomly chosen natural frequency (quenched disorder)  
 and it is stirred 
by Brownian motion. In the limit $N \to \infty$    this model is accurately described
by a (deterministic) Fokker-Planck equation. We study this equation and obtain quantitatively
sharp results in the limit of weak disorder. 
 We show that, in general, even when the natural frequencies have zero mean the oscillators 
synchronize  (for sufficiently strong interaction) around a common rotating phase, whose frequency is sharply estimated. We also establish  the
stability properties of these solutions (in fact, limit cycles). 
These results are obtained by identifying the stable hyperbolic manifold of stationary 
solutions of an associated non disordered model and by exploiting the robustness
of hyperbolic structures under suitable perturbations. 
When the disorder distribution is symmetric the speed vanishes and  there is a
one parameter family of stationary solutions, as pointed out by H.~Sakaguchi  \cite{cf:Sakaguchi}: in this case we provide more precise  stability  estimates.
The methods we use apply beyond the Kuramoto model and we develop here
the case of active rotator models, that is the case in which the dynamics 
of each rotator in absence of interaction and noise  is not simply a rotation.  
  \\[10pt]
  \\[10pt]
\end{abstract}

\maketitle

\section{Introduction}
\label{sec:intro}
\subsection{Collective phenomena in noisy coupled oscillators}
Coupled oscillator models are omnipresent in the scientific literature because 
the emergence of coherent behavior in large families of interacting units that have 
a periodic behavior, that we generically call {\sl oscillators}, is an extremely common 
phenomenon (crickets chirping, fireflies flashing, planets orbiting, neurons firing,...).
It is impossible to properly account for the literature and the various models
proposed for
 this kind of phenomena, but while a precise description of each of the different 
 instances in which synchronization emerges demands specific, possibly very complex, 
 models,    the {\sl Kuramoto
 model} has emerged as capturing some of the
 fundamental aspects of synchronization \cite{cf:acebron}.
 It can be introduced  
 via
 the system of $N$ stochastic differential equations
\begin{equation}
\label{eq:Kmod}
 \dd\gp^\go_j(t)\, =\,
  \go_j \dd t - \frac KN \sum_{i=1}^N \sin(\gp^\go_j(t)-\gp^\go_i(t)) \dd t + \gs \dd B_j(t)\, ,
\end{equation}
for $j=1, \ldots, N$, where
\begin{enumerate}
 \item $\{B_j\}_{j=1, \ldots,  N}$ is a family of standard independent Brownian motions: in physical terms,
 this is a {\sl thermal  noise};
 \item $\{\go_j\}_{j=1\cdots N}$ is a family of independent identically distributed random variables of law $\mu$: they are are the {\sl natural
frequencies} of the oscillators and, in physical terms,
they can be viewed as a   {\sl quenched disorder};
 \item $K$ and  $\gs $ are non-negative parameters, but one should think of them as 
 positive parameters since the cases in which they vanish have only a marginal role 
 in the what follows. 
 \end{enumerate}

\medskip

 The variables $\gp^\go_j$ are meant to be angles (describing the position of rotators on the circle
$\bbS$), so we focus on $\gp^\go_j\;\text{mod}\;2\pi$ and 
\eqref{eq:Kmod} defines, once an initial condition is supplied, a 
diffusion process on $\bbS^N$. Note that if 
$\{ \gp^\go _j (\cdot)\}_{j =1, \ldots, N}$ solves \eqref{eq:Kmod}, also
$\{ \gp^\go _j (\cdot)+ \gp\}_{j =1, \ldots, N}$, with $\gp\in \bbS$,  is a solution: this is the 
rotation symmetry of the system that will repeatedly make surface
in the remainder of the paper.

\medskip 

Some of the main features  \eqref{eq:Kmod} are easily grasped:
each oscillator rotates at its own speed, it is perturbed by independent 
noise and it interacts with all the other oscillators: the interaction tends to
align the rotators.   It may be helpful at this stage to point out 
that if $\mu = \gd_0$, that is the natural frequencies are just zero, then the dynamics is reversible with invariant probability measure that, up to normalization, is
\begin{equation}
\label{eq:mfr}
\exp\left(
\frac{K}{\gs^2}
\sum_{i,j=1}^N \cos \left( \gp_i -\gp_j\right)
\right)  \gl_N (\dd \gp)\, ,
\end{equation}
where $\gl_N$ is the uniform measure on $\bbS^N$. 
The Gibbs measure in \eqref{eq:mfr} is a well known statistical mechanics model --
it is the classical XY spin mean field model or rotator mean field model --
treated analytically in \cite{cf:SFN,cf:Pearce} in the $N \to \infty$ limit. In particular, 
the model exhibits a phase transition at $K=K_c:= 1/\gs^2$, that is effectively 
a {\sl synchronization transition}:  in the $N \to \infty$ limit we have that  for $K\le  K_c$ 
the rotators become independent and 
uniformly distributed over $\bbS$, while for $K>K_c$ the limit measure is obtained by choosing a phase $\theta$ uniformly in $\bbS$ and by choosing the values of the 
phase of each oscillator by drawing it at random following a  
suitable distribution that concentrates around $\theta$.
However, in \cite[Prop.~1.2]{cf:BGP}, it is shown that, unless $\mu=\gd_0$, 
the model is  not reversible (for $\mu$ almost surely all the realization of $\go$)
and one effectively steps into the domain of non-equilibrium statistical mechanics.

Our approach actually relies on a sharp control of the reversible case and works when
the system is not too far from reversibility, that is for weak disorder. Our approach 
actually applies well beyond  \eqref{eq:Kmod}: here we will treat explicitly
the case $\go_j$ is replaced by $U(\gp_j^\go, \go_j)$, that is the natural frequency
$\go_j$ is replaced by a {\sl natural dynamics} that can be substantially different
from one oscillator to another. This model is a disordered version of the active
rotator model considered for example in \cite{cf:shinomoto1986a}.

Since we will focus on $\gs >0$, from now on, for ease of exposition, we set $\gs :=1$. 
 
 \subsection{The Fokker-Planck or McKean-Vlasov limit}
 An efficient way to tackle  \eqref{eq:Kmod} is to consider the empirical probability
 on $\bbS\times \bbR$
 \begin{equation}
 \label{eq:empP}
 \nu_{N, t}^\go (\dd \theta, \dd \go )\, :=\, \frac 1N \sum_{j=1}^N 
 \gd_{(\gp_j^\go(t), \go_j)} (\dd \theta, \dd \go)\, .
 \end{equation}
In fact, in  the  $N\to \infty$ limit, the sequence of measures 
$\{ \nu_{N, t}^\go\}_{N=1,2 , \ldots}$  converges to a limit measure whose density (with respect
to $\gl_1 \otimes \mu$) solves  the nonlinear
Fokker-Planck equation
\begin{equation}
\label{FKP kuramoto disorder}
 \partial_t p_t(\gtta,\go)\, =\, \frac{1}{2} \gD p_t(\gtta,\go) -\partial_\gtta \Big(p_t(\gtta,\go)(\langle
J*p_t\rangle_\mu(\gtta) +\go)\Big),
\end{equation}
where $J(\gtta)=-K\sin(\gtta)$, $\ast$ denotes the convolution and $\langle \cdot\rangle_\mu$ is a notation for the integration with respect to $\mu$, so $\langle J\ast u\rangle_\mu(\gtta) = \intR\intS{J(\gp)
u(\gtta-\varphi, \go)\dd\gp\mu(\dd\go)}$ is the convolution of $J$ and $u$, averaged with respect to the
disorder. 
Here and throughout the whole paper  $\gD$ means $\partial^2 _\gt$.
The Fokker-Planck PDE \eqref{FKP kuramoto disorder} appears repeatedly in the physics and biology literature, see e.g. \cite{cf:acebron,cf:Sakaguchi,cf:StrogatzMirollo}, and a mathematical 
proof (and precise statement) of the result we just stated can be found
in \cite{cf:dPdH,cf:eric}. Notably, in \cite{cf:eric} the result is established 
under the assumption that $\int \vert \go \vert \mu( \dd \go)< \infty$ and emphasis 
is put on the fact that the result holds for almost every realization of the disorder 
sequence $\{ \go_j\}_{j=1,2, \ldots}$. Let us point out that in \eqref{FKP kuramoto disorder}  $\go$ is a one dimensional real variable, while in \eqref{eq:Kmod} the superscript 
$\go$ is a short for the whole sequence of natural frequencies. Since what follows 
is really about  \eqref{FKP kuramoto disorder} this abuse of notation will be of limited impact. 

In Appendix \ref{sec:appendix regularity pt with disorder}, we detail the fact that
 \eqref{FKP kuramoto disorder} generates an evolution
semigroup in  suitable spaces. Here we want to stress that 
 \eqref{FKP kuramoto disorder} can be viewed as a family of coupled
 PDEs, one for each value of $\go$ in the support of $\mu$:  $p_t(\cdot, \go)$
 is the distribution of phases in the population of oscillators with natural
 frequency $\go$.
 
 \subsection{About stationary solutions to  \eqref{FKP kuramoto disorder}}
 Remarkably (\cite{cf:Sakaguchi}, see also \cite{cf:dH}),
 if $\mu$ is symmetric all the 
 stationary solutions to  \eqref{FKP kuramoto disorder} can be written in a 
 semi-explicit way as ${q}(\gtta+\gtta_0, \go)$ ($\gtta_0$ is an arbitrary constant 
 that reflects the rotation symmetry) where
\begin{equation}
\label{eq:qhatom}
{q}(\gtta, \go) \,:=\, \frac{S(\gtta, \go, 2Kr)}{Z(\go,2Kr)}\, ,
\end{equation}
with 
\begin{equation}
\label{eq:Sq}
 S(\gtta, \go, x) \,=\, e^{G(\gtta, \go, x)}\left[ (1-e^{4\pi\go})
\int_{0}^{\gtta}{e^{-G(u,
\go, x)}\dd u} + e^{4\pi\go}\int_{0}^{2\pi}{e^{-G(u, \go, x)} \dd u} \right],
\end{equation}
and $G(u, y, x)= x\cos(u) + 2y u$, $Z(\go,x)= \intS S(\gtta, \go, x) \dd\gtta$ is the 
normalization constant and
$r\in[0,1]$ satisfies the fixed-point relation
\begin{equation}
 \label{eq:fixedpointom}
r\,=\, \Psi^\mu(2Kr), \ \ \  \text{where} \ \ \  \Psi^\mu(x)\,:=\, \intR\frac{\intS\cos(\gtta)S(\gtta, \go,
x)\dd\gtta}{Z(\go,x)}\mu(\dd\go)\, .
\end{equation}
A series of remarks are in order:
\begin{enumerate}
\item $r=0$ solves \eqref{eq:fixedpointom} and this corresponds to the 
fact that $q(\cdot)\equiv\frac1{2\pi}$ is a stationary solution. It is the only
one  as long as $K$ does not exceed 
critical value $K_c$ which is in any case not larger than
\begin{equation}
\label{eq:Ktilde}
\widetilde{K} \, :=\, 
 \left( \int_\bbR \frac{\mu(\dd \go)}{1+4\go^2} \right)^{-1}\, ,
 \end{equation} 
 as one can easily see   by computing (see e.g. \cite{cf:dH}) the  derivative of 
$\Psi^\mu(2K\cdot)$ at the origin and noticing that is larger than one if and only if
$K >\tilde K$ and that  $\Psi^\mu(\cdot) <1$, see Figure~\ref{fig:fixed point Psi}.
\item When \eqref{eq:fixedpointom} admits a fixed point $r>0$, and this 
is certainly the case if $K>\widetilde K$, a nontrivial 
stationary solution is present and in fact, by rotation symmetry, a circle 
of non-trivial stationary solutions. Such solutions correspond to a synchronization phenomenum, since the distribution of the phases is no longer
trivial.
\item As explained in Figure~\ref{fig:fixed point Psi} and its caption, in general there can be
more than one fixed point $r>0$: in absence of disorder there is only one
positive fixed point (when it exists, that is for $K>1$), but this fact is non-trivial
even in this case (see below). Uniqueness  is expected for $\mu$ which is unimodal, but
this has not been established. 
\item While the local stability of $\frac 1{2\pi}$ is understood 
\cite{cf:StrogatzMirollo} and it holds only if $K \le \widetilde{K}$, the stability properties of the non-trivial solutions are a
more delicate issue. 
\end{enumerate}

\begin{figure}
\centering
\subfloat[Function $\Psi^\mu(2K\cdot)$ for $\mu=\delta_0$, $K=2$.]{\includegraphics[width=6cm]{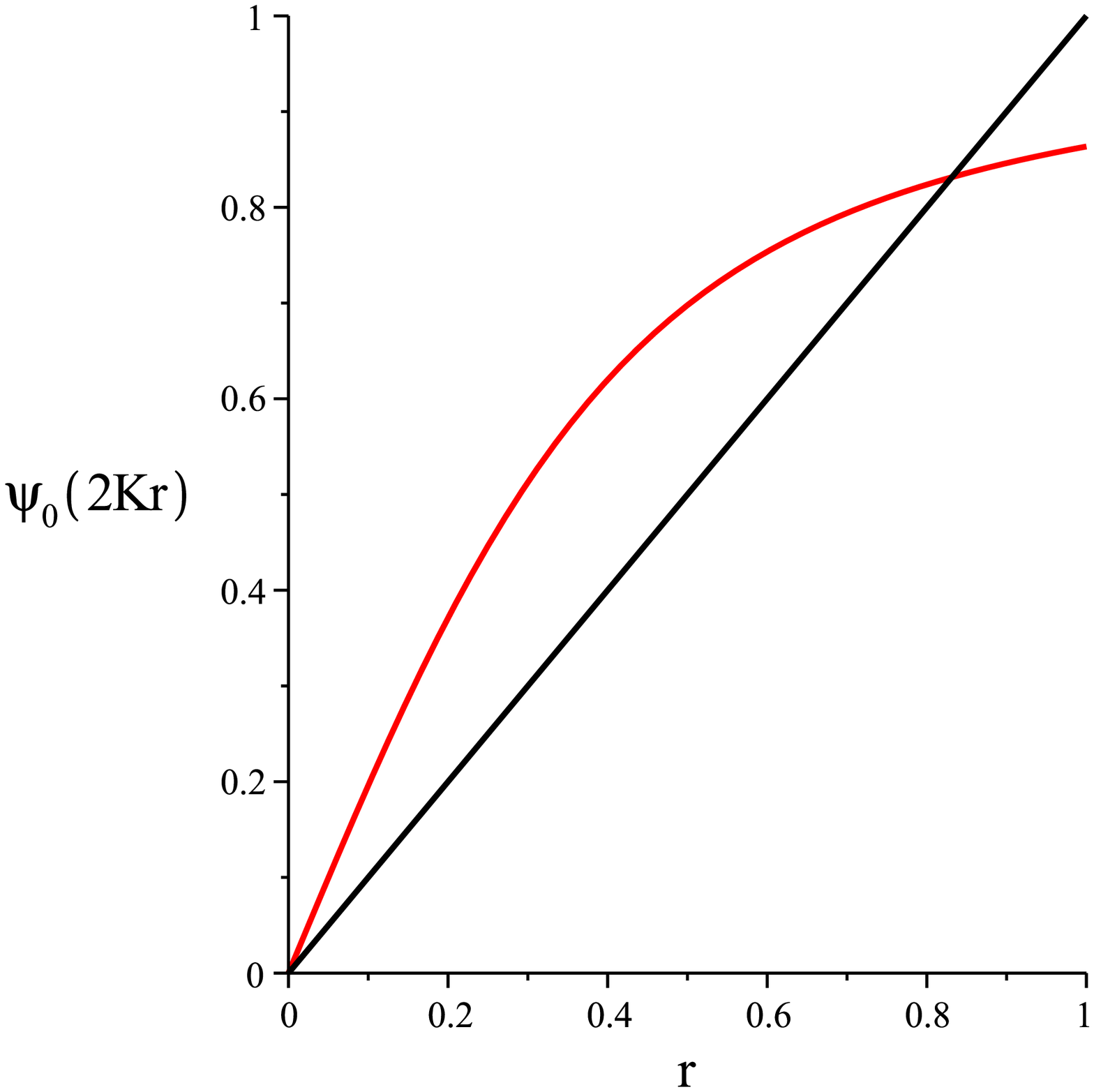}
\label{subfig:fixedpointPsi mu delta0}}
\quad\subfloat[Function $\Psi^\mu(2K\cdot)$ for $\mu=\frac12\left( \gd_{-1} + \gd_1\right)$,
$K=3.5$.]{\includegraphics[width=6cm]{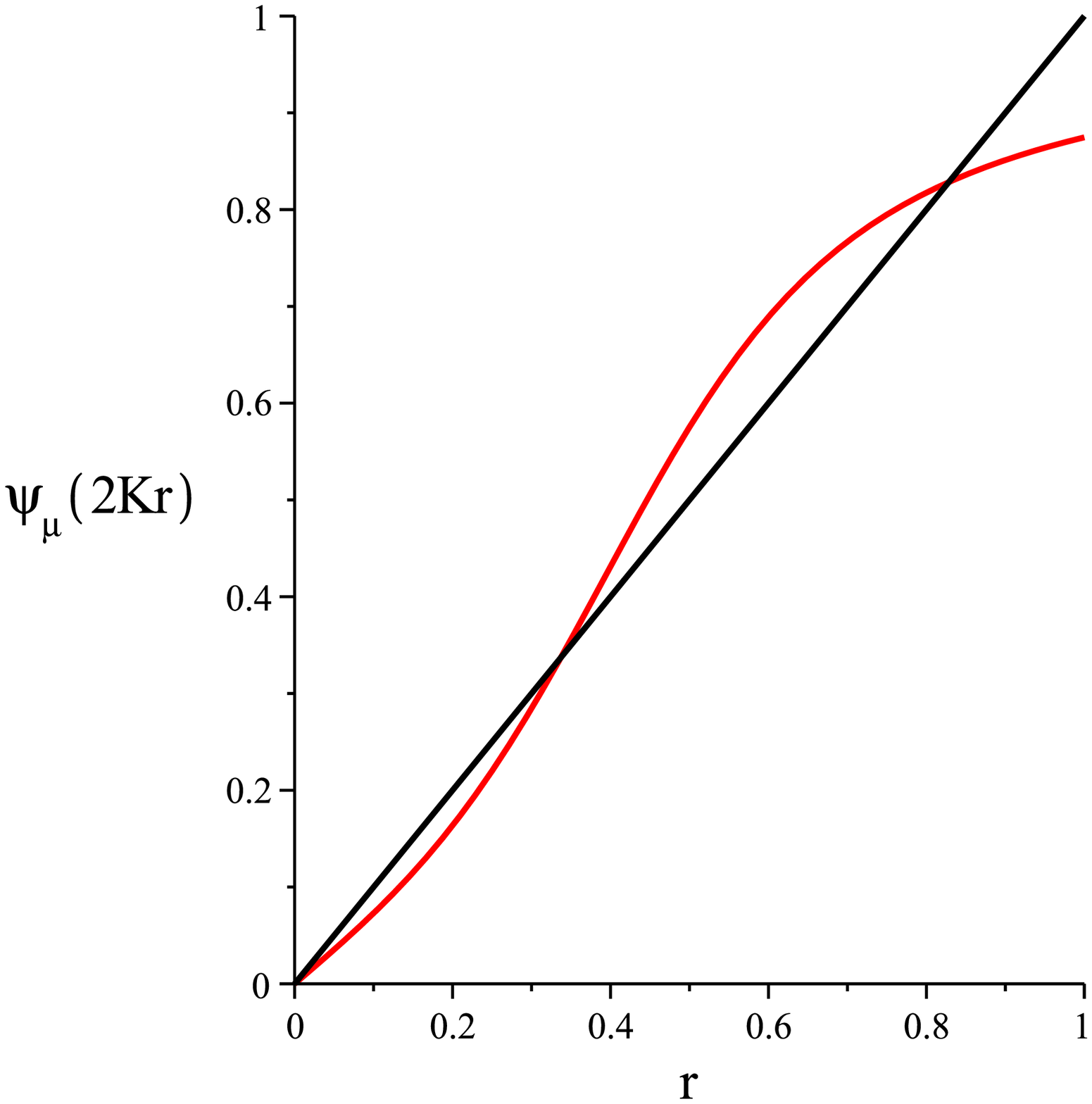}
\label{subfig:fixedpointPsi mu}}
\caption{Plot of the fixed-point function $\Psi^\mu(2K\cdot)$ for two choices of $K$ and $\mu$. 
$\Psi^{\delta_{0}}(\cdot) $ is strictly concave with derivative at the origin equal to $1/2$ (Fig. \ref{subfig:fixedpointPsi mu delta0}) but 
 even for a simple instance of $\mu$ 
(Fig. \ref{subfig:fixedpointPsi mu}) concavity is lost and there are
 several non-trivial fixed-points, each of them corresponding to one circle of non trivial stationary solutions.
Note that in the case of Fig. \ref{subfig:fixedpointPsi mu}, $K < \widetilde{K}=5$ so that the phase transition is not given by
the derivative of $\Psi^\mu(2K\cdot)$ at the origin.}%
\label{fig:fixed point Psi}%
\end{figure}

\subsection{An overview of the results we present}
Here are two natural questions:
\begin{itemize}
\item What are  the stability properties of the   non-trivial stationary solutions?
\item What happens if $\mu$ is not symmetric?  
\end{itemize}

Our work addresses these two questions and provides complete answers for weak disorder. 
The precise set-up of our work is better understood if we remark from now that we can 
assume $m_\go:=\int \go \mu (\dd \go)=0$. In fact, if this is not the case we can map the model
to a model  with $m_\go=0$  by putting ourselves on the frame that rotates with speed $m_\go$,
that is if we consider the diffusion $\{\gp^\go _{j}(t)- m_\go t\}_{j=1, \ldots, N}$. So, we assume henceforth $m_\go=0$
and we  rewrite the natural frequencies as $\gd \go$, with $\gd$ a non-negative parameter. We assume moreover that 
\begin{equation}
\label{eq:suppHyp}
\Supp(\mu)\subseteq [-1, 1]\, .
\end{equation}
In this set-up, \eqref{FKP kuramoto disorder} becomes 
\begin{equation}
\label{FKP kuramoto disorder delta}
 \partial_t p_t^\gd(\gtta,\go)\, =\, \frac{1}{2} \gD p_t^\gd(\gtta,\go) -\partial_\gtta \Big(p_t^\gd(\gtta,\go)(\langle
J*p_t^\gd\rangle_\mu(\gtta) +\gd\go)\Big)\, .
\end{equation}
Note that this leads to  (obvious) changes to \eqref{eq:qhatom}-\eqref{eq:fixedpointom}.
We have introduced this parameterization because the results that we present are for 
small values of $\gd$. 
In particular we are going to show that for any $K>1$, there exists $\gd_0>0$ such that for $\gd \in [0, \gd_0]$
\begin{itemize}
\item there exists a solution $p_t^\gd(\gtta,\go)$ 
to  \eqref{FKP kuramoto disorder delta} of the form $q( \gtta - c_\mu (\gd) t)$, we show that $c_\mu (\gd) =O(\gd^3)$ and we 
actually give an expression
for
$\lim_{\gd \searrow 0}c_\mu (\gd) / \gd^3$: this is a rotating wave (or limit-cycle) for the dynamical system \eqref{FKP kuramoto disorder delta}
and we  establish its stability under perturbations; 
\item when  $\mu$ is  symmetric and $K> \widetilde K$ we show that there is, up to rotation symmetry,  only one non-trivial solution and that it is (linearly and non-linearly) stable.
\end{itemize}

\medskip

The results we obtain are based on the rather good understanding that we have of the case  $\gd=0$ that, 
as we have already explained, is reversible and the corresponding Fokker-Planck PDE is of gradient flow type (e.g. \cite{cf:Otto} and references therein).
These properties have been exploited in \cite{cf:BGP} in order to extract a number of properties of 
the Fokker-Planck PDE (denoted from now on: reversible PDE)
\begin{equation}
\label{eq:revK}
 \partial_t p_t(\gtta)\, =\, \frac{1}{2} \gD p_t(\gtta) -\partial_\gtta \Big(p_t(\gtta)( J*p_t)(\gtta) \Big) \, ,
\end{equation}
and notably the linear stability of the non-trivial stationary solutions. In fact one can find in \cite{cf:BGP}
an analysis of the  evolution operator linearized around the non-trivial stationary solutions. Some of the results 
in \cite{cf:BGP} are recalled in the next section, but they are not directly applicable because 
the $\gd=0$ case that corresponds to what interests us is rather
\begin{equation}
\label{FKP kuramoto disorder without drift}
 \partial_t p_t(\gtta, \go)\, =\, \frac{1}{2} \gD p_t(\gtta, \go) -\partial_\gtta \Big(p_t(\gtta, \go)(\langle J*p_t\rangle_\mu(\gtta) \Big) \, ,
\end{equation}
which we call {\sl non-disordered PDE}.
So the {\sl natural frequencies} have no effective role
beyond separating the various rotators into populations with given natural (ineffective) frequency
that now are just labels. 
But    in order to set-up a proper 
 perturbation procedure we need to control \eqref{FKP kuramoto disorder without drift} and, in particular,
 we need (and establish) a spectral gap inequality for the evolution \eqref{FKP kuramoto disorder without drift} linearized around the non-trivial solutions. 

This spectral analysis is going to be central both for the general and for the symmetric disorder case. In the general set-up we are going to exploit the
{\sl normally hyperbolic structure} \cite{cf:HPS,cf:SellYou} of the manifold of stationary solutions of \eqref{FKP kuramoto disorder without drift}
and the robustness of such structures (like in \cite{cf:GPPP}). In the case of 
symmetric $\mu$ we can get more precise results by 
ad hoc estimates, made possible by the explicit expressions
\eqref{eq:qhatom}-\eqref{eq:fixedpointom}, and use results in the general theory of operators \cite{Pazy1983} and
perturbation theory of self-adjoint operators \cite{Kato1995}. 

The normal hyperbolic manifold approach  allows to treat cases that are substantially more general and notably the case
of 
\begin{equation}
\label{eq:AR}
 \partial_t p_t(\gtta,\go)\, =\, \frac{1}{2} \gD p_t(\gtta,\go) -\partial_\gtta \Big(p_t(\gtta,\go)(\langle
J*p_t\rangle_\mu(\gtta) +\gd U(\theta, \go))\Big)\, ,
\end{equation}
which is the large $N$ limit of \eqref{eq:Kmod} with the term $\go_j \dd t$ replaced by 
$U(\gp^\go_j (t), \go_j) \dd t$, with $U \in C^1( \bbS\times \bbR; \, \bbR)$. 
In this case each oscillator has its own  non-trivial dynamics which may be very different from
the dynamics of other oscillators: consider for example
\begin{equation}
U(\gp, \go)\, =\,b+ \go + a\sin (\gp)\, , \ \ \ a, b \in \bbR \, ,
\end{equation} 
and $\mu$ uniform over $[-1,1]$. For $a \in (-1,1)$ there are some  {\sl active rotators} 
\cite{cf:shinomoto1986a,cf:GPPP}
 that in absence of noise and interaction ($\gs=K=0$) rotate (this happens if $\vert b+\go\vert >\vert a\vert$ and of course
 the direction of rotation depends on the sign of $b+\go$) 
 and others that instead are stuck at a fixed point (this happens if $\vert b+\go\vert \le \vert a\vert$). 
 Our approach allows us to establish that there is a synchronization regime  
 for $K>1$ and $\gd$ small and to describe the dynamics of the system in this regime.
 This is going to be detailed
 in Section~\ref{sec:AR}. 

\medskip

The two questions raised at the beginning of this section have been already repeatedly approached but looking at synchronized solutions as
bifurcation from incoherence. The results are hence for $K$ close to the critical value corresponding to the breakdown of linear stability of $1/2\pi$:
one can find a detailed review of the vast literature on this issue in \cite[Sec. III]{cf:acebron}.
Our results are instead for arbitrary $K>1$, but $\gd$ smaller than $\gd_0(K)$ and of course $\gd_0(K)$ vanishes as $K$ approaches $1$.

\section{Mathematical set-up and main results}
\label{sec:mainresults}
\subsection{The reversible  and the non-disordered PDE}
\label{subsec:organizednondisorderedcase}
We first recall some results about the reversible PDE \eqref{eq:revK}. 
The stationary solutions $q_0(\theta)=q(\theta,0)$  are, up to rotation invariance, given by
\eqref{eq:qhatom}-\eqref{eq:fixedpointom}, but formulas get simpler, namely
\begin{equation}
\label{eq:defstationarysolution nodisorder}
 q_0(\gtta)\, =\, \frac 1{Z_0(2Kr_0)} \exp(2Kr_0 \cos(\gtta))\, ,
\end{equation}
where $Z_0(x):= Z(0, x)^\frac 12$ and 
this time we have the more explicit expression $Z_0(x)\, =\, \intS e^{x\cos(\gtta)} \dd\gtta = 2\pi I_0(x)$ is the normalization
constant and $r_0$ is a solution of
the fixed-point problem
\begin{equation}
\label{eq:deffixedpoint nodisorder}
 r_0\, =\, \Psi_0 (2Kr_0) \qquad \text{where} \qquad \Psi_0(x)\, :=\, \frac{I_1(x)}{I_0(x)}\, ,
\end{equation}
where we used standard notations for the modified Bessel functions
\begin{equation}
 \label{def bessel}
I_i(x)\, =\, \frac 1{2\pi}\int_{\bbS} (\cos(\gtta))^i\exp(x\cos(\gtta))\dd\gtta\, \qquad i=0,1\, .
\end{equation}
The mapping $\Psi_0$ is increasing, concave (see \cite{cf:Pearce}) and with derivative at $0$ equal to $\frac 12$. Consequently if
$K\leq
1$, $r_0=0$ is
the unique solution of the fixed-point problem, and $q(\cdot)\equiv\frac{1}{2\pi}$ is the only stationary solution of \eqref{eq:revK}. 
If $K>1$, we get in addition a circle (because of the rotation invariance) of nontrivial stationary solutions
\begin{equation}
 M_{\text{rev}}\, :=\, \{q_{\psi, 0}(\cdot):=q_0(\cdot-\psi):\,  \psi\in\bbS\} \qquad \text{with} \qquad q_0(\gtta)\,
:=\,\frac{\exp(2Kr_0\cos(\gtta))}{\int_{\bbS }\exp(2Kr_0\cos(\gtta))}
\end{equation}
where $r_0=r_0(K)$ is the unique non trivial fixed-point \eqref{eq:deffixedpoint nodisorder}.

\medskip

Let us now focus on the non-disordered PDE \eqref{FKP kuramoto disorder without drift}
and let us insist on the fact that we are interested in solutions such
that $\gp^\gd_t(\cdot, \go)$ is a probability density.
Observe then that if $q(\gtta,\go)$ is a stationary solution of \eqref{FKP kuramoto disorder without drift}, we see (Appendix
\ref{sec:appendix regularity pt with disorder}) that $q$ is
$C^\infty$ with respect to $\gtta$ and that $\langle q \rangle_\mu$ is a stationary solution for \eqref{eq:revK}. So there
exists $\psi\in\bbS$ such that $\langle q\rangle_\mu = q_\psi$ and a short computation leads to
\begin{equation}
\langle J*q \rangle_\mu(\gtta)\, =\, -K\sin(\gtta-\psi)\, ,
\end{equation}
and, since 
$\intS q(\gtta,\go) \dd \gtta =1$ for almost all $\go$, we obtain
that $q(\cdot,\go)=q_{ \psi}(\cdot)$ for almost all $\go$. In conclusion, with some abuse of notation, we can say  the stationary solutions of \eqref{eq:revK} and
\eqref{FKP kuramoto disorder without drift} are the same: of course in  the second case 
the function space includes the dependence on $\go$, so 
we choose a different notation, that is $M_0$,  for the corresponding
circle of non-trivial stationary solutions.

\medskip

An important issue for us is the stability of $M_0$ (for its existence  we are assuming $K>1$) and for this
we denote by $A$ the linearized evolution operator of
\eqref{FKP kuramoto disorder without drift} around $q_0$ 
\begin{equation}
\label{def A}
Au(\gtta,\go)\, :=\, \frac 12 \gD u(\gtta,\go) - \partial_\gtta \Big(q_0(\gtta) \langle J*u \rangle_\mu(\gtta) + u(\gtta,\go)
J*q_0(\gtta)\Big)
\end{equation}
with domain
\begin{equation}
\label{eq:defdomain A}
\cD(A)\, :=\, \left\{u\in C^2(\bbS \times\bbR,\bbR):\, \intS u(\gtta,\go)\dd\gtta=0\,  \text{ for all }\go\right\}\, .
\end{equation}
For any smooth positive function $k:\bbS \mapsto \bbR$, we introduce the Hilbert space $H^{-1}_{k,\mu}$ defined by the closure of
$\cD(A)$ for the norm $\Vert \cdot \Vert_{-1,k,\mu}$ associated with the scalar product
\begin{equation}
\label{def scalar product sobolev with wage and mu}  
\langle u,\, v \,\rangle_{-1,k,\mu}\, :=\, \intSR \frac{\cU(\gtta,\go)\cV(\gtta,\go)}{k(\gtta)}\dd\gtta \mu(\dd\go),
\end{equation}
where $\go$ a.s., $\cU(\cdot,\go)$ is the primitive of $u(\cdot,\go)$ such that
$\int_{\bbS }\frac{\cU(\gtta,\go)}{k(\gtta)}\dd\gtta =0$, and $\cV(\cdot,\go)$ is defined in the analogous  fashion. Let us remark (see
\cite[Sec.~2]{cf:GPPP})
immediately that
\begin{equation}
\Vert u \Vert_{-1,k_1,\mu}^2\,  \leq \, \frac{\Vert k_2 \Vert_\infty}{ \Vert k_1 \Vert_\infty}\Vert u \Vert_{-1,k_2,\mu}^2\, ,
\end{equation}
so that all the norms we have introduced are equivalent. For the case $k(\cdot) \equiv
1$ we use the notations $H^{-1}_\mu$ and
$\Vert\cdot\Vert_{-1,\mu}$.
We will prove the following result, which is just technical, but it will be of help
to understand our main results:

\medskip

\begin{proposition}
\label{th:spectral gap A} 
 A is essentially self-adjoint in $H^{-1}_{{q_0},\mu}$. Moreover the spectrum lies in $(-\infty,0]$, $0$ is a
simple eigenvalue, with eigenspace spanned by $\partial_\gtta q_0$, 
and there is a spectral gap, that is 
the distance $\gap$ between $0$ and the rest of the spectrum is positive.
\end{proposition}
\medskip

The proof of this result builds on \cite[Th.~1.8]{cf:BGP} that deals with the
reversible case and the (lower) bound 
on the spectral gap $\gap$ that we obtain coincides with the quantity $\gl(K)$
in \cite[Th.~1.8]{cf:BGP} (this bound  can be improved
 as explained in in \cite[Sec.~2.5]{cf:BGP} and sharp estimates on the spectral gap
 can be obtained in the limit $K \searrow 1$ and $K \nearrow \infty$).
 For the reversible evolution, the
linear operator $L_{q_0}$ is defined by
\begin{equation}
 \label{def Lq0}
L_{q_0}u(\gtta)\, :=\, \frac 12 \gD u(\gtta) -\partial_\gtta\Big( q_0(\gtta)J*u(\gtta)+u(\gtta)J*q_0(\gtta)\Big),
\end{equation}
with domain $D(L_{q_0})$ given by the $C^2(\bbS ,\bbR)$ functions with zero integral.



\subsection{Synchronization: the main result without symmetry assumption}
Proposition~\ref{th:spectral gap A} is a key ingredient for our main results and 
  the functional space $H^{-1}$ appears in it, but an important role
is played also by $L^2(\gl \otimes \mu)$, $\gl$ is the Haar measure on $\bbS$,
whose norm is denoted by $\Vert \, \cdot \, \Vert_{2, \mu}$. For $C>0$ and $M \subset L^2(\gl \otimes \mu)$ we set
$\cN_{2, \mu}(M, C):= \{u: \,$there exists $v\in M$ such that $\Vert u-v \Vert_{2, \mu} \le C\}$.
In the statement below $q\in M_0$ is the element of the manifold such that $q(\cdot,\go)=
q_0(\cdot)$, cf. \eqref{eq:defstationarysolution nodisorder}, with $r_0(K)>0$ (hence $K>1$).

\medskip

\begin{theorem}
\label{th:expansion speed}
 For every $K>1$
there exists $\gd_0=\gd_0(K)>0$ such that for $\vert \gd \vert \le \gd_0$
there exists $\widetilde q _\gd \in  L^2(\gl \otimes \mu)$, satisfying
$\Vert \widetilde q _\gd -q \Vert_{2, \mu}=O(\gd)$ and a value 
$c_\mu (\gd) \in \bbR$ such that if we set
\begin{equation}
\label{eq:main1.1}
q^{(\psi)}_t (\theta, \go)\, :=\, \widetilde q _\gd(\theta -c_\mu (\gd) t-\psi)\, ,
\end{equation}
then $q^{(0)}_t$ solves \eqref{FKP kuramoto disorder delta}. Moreover 
\begin{enumerate}
\item
the family of solutions $\{q^{(\psi)}_\cdot\}_\psi$
is stable in the sense that there exist two positive constants $\gb=\gb(K)$ and $C=C(K)$
such that if $p_0^\gd \in \cN_{2, \mu}(M_0,\gd)$, and $\int_\bbS p_0^\gd
(\theta, \go) \dd \theta=0$ $\mu(\dd \go)$-a.s., then there exists
$\psi_0\in \bbS$ such that for all $t\ge 0$
\begin{equation}
\Vert q_t^{(\psi_0)}-p_t^\gd \Vert_{2, \mu}\, \le \, 2 C\exp(-\gb t)\, .
\end{equation}
\item
we have 
\begin{equation}
\label{eq:main1.2}
 c_\mu(\gd)\,  =\, 
  \gd^3
   \frac{
  \left\langle \go\partial_\gtta n^{(2)} , \partial_\theta q_0  \right \rangle_{-1, {q_0},\mu} }
  { \left\langle \partial_\theta q_0 ,\partial_\theta q_0
  \right \rangle_{-1, {q_0}} } + O(\gd^5)\, , 
\end{equation}
where $n^{(2)}$ is the unique solution of
\begin{equation}
An^{(2)}\, =\, 
\go
\partial_\gtta n^{(1)} \quad\text{ and } \quad  \left\langle n^{(2)}, \partial_\theta q_0 \right \rangle_{-1, {q_0},\mu}
\, =\, 0\, ,
\end{equation}
and $n^{(1)}$ is the unique solution of
\begin{equation}
A n^{(1)}\, =\, 
\go \, \partial_\theta q_0 \quad\text{ and } \quad \crosqmu {n^{(1)}}{ \partial_\theta q_0 }\, =\, 0\, .
\end{equation}
\end{enumerate}
\end{theorem}

\medskip

In the  proof  of Theorem~\ref{th:expansion speed} one finds also further estimates, in
particular (see \eqref{eq:n_2}) that one has
\begin{equation}
\label{eq:qgddevel}
\tilde q_\gd \, =\, q_0 + \gd n^{(1)} + \gd^2 n^{(2)} + O_{L^2}(\gd^3)\, .
\end{equation}
Actually, see Remark~\ref{rem:pa}, the argument of proof can be pushed  farther
to obtain arbitrarily many terms in development \eqref{eq:qgddevel}, as well as in 
\begin{equation}
 c_\mu (\gd)\, =\, c_3 \gd^3+ c_5 \gd^5 +\ldots \,  .
\end{equation}
In Table~\ref{tab:1} we report a comparison between 
the $c_\mu(\gd)$ obtained by solving numerically \eqref{FKP kuramoto disorder delta}
and by evaluating the leading order $c_3$, i.e. by using \eqref{eq:main1.2}.

\begin{table}[h!]

\begin{center}
  
   \begin{tabular}{ | c | c | c | c | }

    \hline \phantom{m}
     $\gd$  & $K=2$                    & $ K=1.5 $                 & $ K=1.1 $                  \\ \hline
     $0.5$  & $-1.56300 \cdot 10^{-2}$ & $-8.59626 \cdot 10^{-2}$  & $-3.01064 \cdot 10^{-1}$   \\ \hline
     $0.1$  & $-1.23998 \cdot 10^{-2}$ & $-6.84835 \cdot 10^{-2}$  & $-2.72117 \cdot 10^{-1}$   \\ \hline
     $0.05$ & $-1.23072 \cdot 10^{-2}$ & $-6.79553 \cdot 10^{-2}$  & $-2.69460 \cdot 10^{-1}$   \\ \hline
     $0.01$ & $-1.22776 \cdot 10^{-2}$ & $-6.77921 \cdot 10^{-2}$  & $-2.68603 \cdot 10^{-1}$   \\ \hline
     $0.005$& $-1.22767 \cdot 10^{-2}$ & $-6.77869 \cdot 10^{-2}$  & $-2.68576 \cdot 10^{-1}$   \\ \hline
     \multicolumn{4}{c}{} \\ \hline
     $c_3$  & $-1.22764 \cdot 10^{-2}$ & $-6.77851 \cdot 10^{-2}$  & $-2.68567 \cdot 10^{-1}$   \\ \hline

   \end{tabular}

 \caption{\label{tab:1} For the case
$
 \mu= p \gd_{1-p} + (1-p) \gd_{-p}$, $p=0.2$, we have computed (numerically)
  $c_\mu(\gd)/\gd^3$ for three values of $K$
 and five values of $\gd$. In the last line we report the value 
 $c_3= \lim_{\gd\searrow 0} c_\mu(\gd)/\gd^3$ that one obtains by using \eqref{eq:main1.2}. }

\end{center}

\end{table}

\subsection{Symmetric disorder case}

Let us focus on the case in which the 
 distribution of the disorder $\mu$ is symmetric. 
In this case, at least for small disorder,
Theorem~\ref{th:expansion speed} is just telling us that the leading order 
in the development for the speed $c_\mu(\gd)$ is  zero: one can actually work harder and show
that such a  development yields zero terms  to all orders. 
In reality  in this case we already know, see  \eqref{eq:qhatom}-\eqref{eq:fixedpointom},
that for $K$ sufficiently large there is at least a non-trivial stationary profile, hence, by
rotation symmetry, at least one whole circle of stationary solutions. 
Actually, we can show that for $\gd$ small there is just one circle,
that we call $M_\gd$, of non-trivial stationary solutions
and this circle converges to $M_0$ as $\gd \searrow 0$ (in $C^j$, for every $j$) so 
the rotating solutions found in Theorem~\ref{th:expansion speed} must be the stationary solutions
in $M_\gd$. 

In order to be precise about this issue, we point out that
 \eqref{eq:qhatom}-\eqref{eq:fixedpointom} are written for 
\eqref{FKP kuramoto disorder}
while we work rather with
\eqref{FKP kuramoto disorder delta}. The changes are obvious, but we introduce a notation
for the analog of \eqref{eq:fixedpointom}:
\begin{equation}
 \label{eq:deffixedpoint disorder}
r_\gd\,=\, \Psi_\gd^\mu(2Kr_\gd), \qquad \text{where,}\, \Psi_\gd^\mu(x)\,:=\, \intR\frac{\intS\cos(\gtta)S(\gtta, \gd\go,
x)\dd\gtta}{Z(\gd\go, x)}\mu(\dd\go)\, .
\end{equation}

\medskip

\begin{lemma}
 \label{th: Psimu concave}
For all $K_{\min}<K_{\max}$, there exists $\gd_1=\gd_1(K_{\min}, K_{\max})>0$ such that, for all $0<K_{\min}<K<K_{\max}$
and  all $\gd\leq\gd_1$ the function $\Psi_\gd^\mu$ is strictly concave on $[0,1]$. Therefore for 
\eqref{eq:fixedpointom} has only a positive solution $r_\gd=r_\gd(K, \mu)$. Moreover $
\lim_{\gd \searrow 0} r_\gd= r_0$.
\end{lemma}
\medskip

\medskip
We point out that in spite of the fact that $\Psi^\mu$ is explicit (cf. \eqref{eq:deffixedpoint nodisorder}), it is not so straightforward to
show that it is concave. We show that $\Psi_\gd^\mu$ remains strictly concave for a small $\gd$ via a perturbation argument. But the conjecture
(see \cite{cf:dH} and \cite{cf:dPdH}) that $\Psi^\mu$ is strictly concave for unimodal distributions $\mu$ is still an open issue.

\begin{rem}
\rm
 A direct computation  shows that the derivative of $\Psi^\mu_\gd$ at the origin is $1/(2\widetilde K _\gd)$, for 
$\widetilde K _\gd\,:=\,
\left(\intR \frac{\mu(\dd\go)}{1+4\gd^2\go^2}\right)^{-1}$
(of course $\widetilde K_1$ coincides with $\widetilde K$,
introduced in \eqref{eq:Ktilde}). 
Under the hypothesis of Lemma \ref{th: Psimu concave}, one therefore  
sees that there is a synchronization transition at $K=\widetilde K _\gd$ in the sense
that for $K\le \widetilde K _\gd$ the only stationary solution is $\frac 1{2\pi}$
while for $K>\widetilde K _\gd$ also the manifold of non-trivial stationary solutions
appears (and there is no other stationary solution).
\end{rem}
\medskip

Theorem~\ref{th:expansion speed} provides a  
  stability statement for $M_\gd$. This result can be sharpened  and for this let us introduce
the linear operator
\begin{equation}
\label{eq:defLqmu}
L^\go _{{q}}u(\gtta, \go)\, :=\,  \frac 12 \Delta u(\gtta, \go) - \partial_{\gtta} \left(
u(\gtta, \go)\left(
\langle J \ast {q}\rangle_\mu(\gtta) + \gd\go\right) + {q}(\gtta, \gd\go) \langle J\ast u\rangle_\mu(\gtta)
\right),
\end{equation}
The domain $\cD(L^\go_{{q}})$ of the operator $L^\go_{{q}}$ is chosen to be the same as for  $A$, cf. \eqref{eq:defdomain A}.
\medskip

We place ourselves within the framework of Lemma \ref{th: Psimu concave}, in the sense that $\delta$ is small enough to ensure
the uniqueness of a non-trivial stationary solution
(of course existence requires $K> \tilde K_\gd$ and this is implied 
by $K>1$ if $\gd$ is sufficiently small). 
We prove a number of properties of the linear operator \eqref{eq:defLqmu}, saying notably that
it has a simple eigenvalue at zero and 
the rest of spectrum  is at a positive distance from zero and it is  in a cone in that lies in the negative complex half plane. We summarize in the next statement the qualitative features 
of our results on $L_q^\go$, but what we really prove are quantitative explicit estimates:
the interested reader finds them in Section~\ref{sec:sym}. 
\medskip

\begin{theorem}
\label{th:spectral prop L disorder}
The operator $L^\go_{{q}}$ has the following spectral properties:
$0$ is a simple eigenvalue for $L^\go_{{q}}$, with eigenspace spanned by $(\gtta, \go)\mapsto q'(\gtta, \go)$.
Moreover, for all $K>1$, $\rho\in(0,1)$, $\alpha\in(0,\pi/2)$, there exists $\gd_2=\delta_2(K, \rho, \alpha)$ such that for all $0\leq
\delta\leq \delta_2$, the
following is true:
\begin{itemize}
\item $L^\go_{{q}}$ is closable  and its closure has the same domain as the domain of the self-adjoint extension of $A$;
\item The spectrum of $L^\go_{{q}}$ lies in a cone $C_{\alpha}$ with vertex $0$ and angle $\alpha$
\begin{equation}
 C_{\alpha}\, :=\, \ens{\lambda\in\bbC}{\frac{\pi}{2} + \alpha\leq
\arg(\lambda)\leq \frac{3\pi}{2}-\alpha}\subseteq \ens{z\in\bC}{\Re(z)\leq 0}\, ;
\end{equation}
\item There exists $\alpha'\in(0, \frac\pi2)$ such that $L^\go_{{q}}$ is the infinitesimal generator of an analytic semi-group defined on a
sector $\{\lambda\in\bbC,\, |\arg(\lambda)|< \alpha'\}$;
\item The distance between $0$ and the rest of the spectrum is strictly positive and is at least equal to
$\rho\gap$, where $\gap$ is the spectral gap of the operator $A$ introduced in Proposition \ref{th:spectral
gap A}.
\end{itemize}
\end{theorem}

\subsection{Organization of remainder of the paper}
In Section~\ref{sec:hyperbolic structure}
we introduce the notion of stable normally hyperbolic manifold,
we recall its robustness properties, 
and show that $M_0$ is in this class of manifolds. 
The essential ingredient is Proposition~\ref{th:spectral gap A} 
that, directly or indirectly, plays a role in each subsequent section. 
Section~\ref{sec:hyperbolic structure} is also devoted to 
the proof of Proposition~\ref{th:spectral gap A}.
The proof of 
Theorem~\ref{th:expansion speed} is then completed 
in Section~\ref{sec:pa}, that is mainly devoted to perturbation
arguments. The case of the active rotators is treated in 
Section~\ref{sec:AR}, while 
Section~\ref{sec:sym} deals with the case symmetric disorder distribution and,
notably, with the proof of Theorem~\ref{th:spectral prop L disorder}
and of a number of related quantitative estimates.

\section{Hyperbolic structures and periodic solutions}
\label{sec:hyperbolic structure}
In this section we present the arguments proving the existence of the periodic solution of Theorem \ref{th:expansion speed}. We rely on the fact
that the circle of stationary solutions $M_0$ is a stable normally hyperbolic manifold, and on the robustness of this kind of structure : adding
the perturbation term $-\gd\partial_\gtta(p_t(\gtta,\go)\go)$ in \eqref{FKP kuramoto disorder without drift}, this manifold $M_0$ is deformed into
another manifold $M_\gd$, and thanks to the rotation invariance of the problem, $M_\gd$ is a circle too. The spectral gap of operator $A$
(Property \ref{th:spectral gap A}) which induces the hyperbolic property of $M_0$ is proved at the end of this section. 

\subsection{Stable normally hyperbolic manifolds}
We start by quickly reviewing
 the notion of of stable normally hyperbolic manifold (SNHM).
%
%
The evolution of \eqref{FKP kuramoto disorder delta} will be studied in the space $X^1_\mu$ defined by
\begin{equation}
 X^1_\mu\, :=\, \left\{u\in L^2(\gl\otimes\mu),\, \int_{\bbS }u(\gtta,\go)\dd\gtta=1\quad \go\text{ a.s.}\right\}
\end{equation}
where $\gl$ denotes the Lebesgue measure on $\bbS $. This is made possible by the conservative character of the dynamics. The
$L^2$-norm with respect to the measure $\gl\otimes\mu$ will be denoted by $\Vert \cdot\Vert_{2,\mu}$. We will also use the space
$X^0_\mu$ defined by
\begin{equation}
 X^0_\mu\, :=\, \left\{u\in L^2(\gl\otimes\mu),\, \int_{\bbS }u(\gtta,\go)\dd\gtta=0\quad \go\text{ a.s.}\right\}\, .
\end{equation}

To define a SNHM, we need a dynamics: we have in mind \eqref{FKP kuramoto disorder delta}
 but for the moment let us just think of an evolution semigroup in $X^1_\mu$ that gives rise to $\{ u_t\}_{t\ge 0}$, with $u_0=u$,
to which we can associate a linear evolution semigroup $\{\Phi(u, t)\}_{t \ge 0}$ in $X^0_\mu$, satisfying  
 $\partial_t  \Phi(u, t)v =L(t)  \Phi(u, t)v$ and $\Phi(u, 0)v=v$, where $L(t)$ is the operator
 obtained by linearizing the evolution around $u_t$. 

 For us a SNHM $M\subset X^1_\mu$ (in reality we are interested only
 in $1$-dimensional manifolds, that is curves, but at this stage this does not really play a role)
 of characteristics $\gl_1$, $\gl_2$ ($0\le \gl_1< \gl_2$) and $C>0$ is  
 a $C^1$ compact connected manifold which is invariant under the dynamics 
 and for every $u \in M$ there exists a projection $P^o(u)$ on the 
 tangent space of $M$ at $u$, that is $\cR(P^o(u))=:T_uM$, which, for $v \in L^2_0$, satisfies the following properties:
 \begin{enumerate}
 \item  for every  $t\ge 0$ we have 
 \begin{equation}
 \Phi(u,t) P^o (u_0)v\, =\, P^o (u_t)\Phi(u,t) v\, ,
 \end{equation}
 \item we have  
 \begin{equation}
 \label{eq:shyp1}
 \Vert \Phi(u,t) P^o(u_0) v\Vert_{2,\mu} \, \le \, C \exp( \gl_1 t)   
 \Vert v \Vert_{2,\mu}\, ,
 \end{equation}
 and, for 
 $P^s\, :=\, 1-P^o$, we have 
 \begin{equation}
 \label{eq:shyp2}
 \Vert \Phi(u,t) P^s(u_0) v\Vert_{2,\mu} \, \le \, C \exp( -\gl_2 t)   
 \Vert v \Vert_{2,\mu}\, ,
 \end{equation}
 for every $t\ge 0$;
 \item 
 there exists a negative continuation of 
 the dynamics  $\{ u_t\}_{t \le 0}$ and of the linearized
 semigroup
 $\{ \Phi(u, t) P^o(u_0) v \}_{t \le 0}$ and for any such continuation
 we have
 \begin{equation}
 \label{eq:shyp3}
 \Vert \Phi(u,t) P^o(u_0) v\Vert_{2,\mu} \, \le \, C \exp( -\gl_1 t)   
 \Vert v \Vert_{2,\mu}\, ,
 \end{equation}
 for $t \le 0$.
 \end{enumerate}

\medskip

\subsection{$M_0$ is a SNHM}
First of all: the dynamics on $M_0$ is trivial. For $q_\psi\in M_0$, the projection $P^o_{q_\psi}$ on the tangent space is the
projection on the subspace spanned by $q^\prime_\psi$:
\begin{equation}
 \label{def P^o}
P^o_{q_\psi} u\, =\, \frac{\left\langle u,q^\prime_\psi\right\rangle_{-1,q_\psi,\mu}}{\left\langle q^\prime_\psi,q^\prime_\psi\right\rangle_{-1,q_\psi}}q^\prime_\psi
\end{equation}
and since the dynamic on the manifold is trivial, we are allowed to choose for the parameters $\gl_1=0$ and $\gl_2=\gap$
(where we recall that $\gap$ is given by Proposition \ref{th:spectral gap A}).

\medskip

We are in the same situation as in \cite{cf:GPPP}. For a suitable perturbation and if $\gd$ is small enough, the circle $M_0$ is
smoothly
transformed into another SNHM $M_\gd$, which is close to $M_0$. The proof is the same as in \cite[Sec.~5]{cf:GPPP}, which, in turn builds 
on results in \cite{cf:SellYou}): the spaces we are working in are  more general since we have to deal with the disorder. Here suitable
perturbation means being an element of $C^1(X^0_\mu,H^{-1}_{\mu})$, but it is clearly the case for the perturbation $u\mapsto
-\gd\, \go\, \partial_\gtta u$ when $\mu$ is of compact support. The following theorem works for all $C^1(X^0_\mu,H^{-1}_{\mu})$
perturbations:

\medskip

\begin{theorem}
\label{th:M} 
\cite[Sec.~5]{cf:GPPP}
For every $K>1$ there exists $\gd_0>0$ 
such that if $\gd \in [0,  \gd_0]$ there exists a stable  normally hyperbolic  
manifold  $M_\gd$ in $X^1_\mu$ for the perturbed equation \eqref{FKP kuramoto disorder delta}. Moreover 
we can write 
\begin{equation}
\label{eq:M}
M_\gd \, =\, \left\{q_\psi + \phi_\gd \left(q_\psi\right):\, \psi  \in \bbS \right\}\, ,
\end{equation}
for a suitable  function $\phi_\gd\in C^1(M_0,X^0_\mu)$ with the properties that 
\begin{itemize}
\item $\phi_\gd (q) \in \cR (A)$;
\item 
there exists 
$C>0$ such that 
$\sup_{\psi }(\Vert \phi _\gd \left(q_\psi\right)\Vert_{2,\mu}+
\Vert \partial_\psi\phi_\gd(q_\psi)\Vert _{2,\mu}) \le C \gd$.
\end{itemize}
\end{theorem}

\medskip
\begin{rem}\label{rem:un}
\rm
A byproduct of the proof in \cite[Sec.~5]{cf:GPPP} is also that
 $M_\gd$ is the unique invariant manifold in a $L^2(\gl,\mu)$-neighborhood of $M_0$. So in the case of
\eqref{FKP kuramoto disorder delta}, thanks to the symmetry of the problem
that tells us that any rotation of $M_\gd$ is still a invariant manifold, $M_\gd$ is in fact a circle, and that the
dynamics on this circle is a traveling wave of constant (possibly zero) speed $c_\mu(\gd)$. So the invariant manifold we get for \eqref{FKP kuramoto disorder delta} is even $C^\infty$. In
this sense, when dealing with \eqref{FKP kuramoto disorder delta}, we are using 
only part of the strength of Theorem~\ref{th:M}. Of course  this symmetry argument does not
apply when dealing with \eqref{eq:AR}.
\end{rem}
\medskip

\begin{rem}\label{rem:stability}
\rm
Theorem \ref{th:M} addresses the existence and the linear stability of the manifold $M_\gd$. 
The non-linear stability statement in Theorem~\ref{th:expansion speed}(1) follows from 
Theorem~\ref{th:M} 
combined with  \cite[Theorem 8.1.1]{cf:Henry},
when the dynamics is periodic with non zero speed on $M_\gd$. If $M_\gd$ is a manifold of stationary points, the argument for the non-linear stability follows by repeating the argument
in
 \cite[Th. 4.8]{cf:GPP}, where  the non-disordered case is treated.
 \end{rem}

\medskip
We now prove Proposition~\ref{th:spectral gap A} and thus that $M_0$ is a SNHM. 

\medskip

\subsection{The spectral gap estimate (proof of Proposition \ref{th:spectral gap A})}
We start by remarking that $A$ is symmetric
for the scalar product $\crosqmu{\cdot}{\cdot}$ (recall \eqref{def scalar product sobolev with wage and mu}). In fact, for $u$ and
$v$ in $\cD(A)$, a short computation  gives (in the following we use the notation $u^\prime(\gtta,\go)=\partial_\gtta u(\gtta,\go)$)
\begin{align}
\crosqmu{v}{Au} & \, =\, \intSR
\left[\frac{\cV(\theta,\go)}{q_0(\theta)}\left(\frac{ u^\prime(\theta,\go)}{2}-u(\theta,\go)J*q_0(\theta)-q_0(\theta)\langle J*u
\rangle_\mu(\gtta) \right) \right]\dd \theta\dd\mu\\ \notag
                 & \, =\, -\frac 12 \intSR \frac{u(\theta,\go)v(\theta,\go)}{q_0(\theta)}\dd \theta\dd\mu+\int_{\bbR}\int_{(\bbS )^2}
v(\theta,\go)\tilde J*u(\theta,\go')\dd\theta \dd\mu \otimes\mu\, ,
\end{align}
where $\tilde J(\gtta)=K\cos(\gtta)$. We now  first prove an inequality  for $A$ 
that is stronger than the spectral gap inequality and then deduce that $A$
is (essentially) self-adjoint. We define
the two following scalar products, which were used for the non-disordered case in \cite{cf:BGP}:
\begin{equation}
 \label{eq:defNorm-1 non-des}
\crosq{u}{v}\, :=\, \intS \frac{\cU(\gtta)\cV(\gtta)}{q_{0}(\gtta)}\dd\gtta \, ,
\end{equation}
where $\cU(\cdot)$ is the primitive of $u(\cdot)$ such that $\intS \frac{\cU(\gtta)}{q_{0}(\gtta)}\dd\gtta=0$ and
\begin{equation}
 \label{eq:defNorm2q non-des}
\cro{u}{v}_{2,q_0}\, :=\, \intS \frac{u(\gtta)v(\gtta)}{q_0(\gtta)}\dd\gtta\, .
\end{equation}
We denote the closures of $\cD(L_0)$ for these scalar products respectively by $H^{-1}_{{q_0}}$ and $L^2_{{q_0}}$. In the disordered
case, $L^2_{{q_0}}$ corresponds to the space $L^2_{{q_0},\mu}$, which we define by the closure of $\cD(A)$ with respect to the
norm $\Vert \cdot\Vert_{2,{q_0},\mu}$ associated with the scalar product
\begin{equation}
\label{def scalar product L2 with wage and mu}
 \croLqmu{u}{v}\, :=\, \intR \intS \frac{u(\gtta,\go)v(\gtta,\go)}{q_0(\gtta)} \dd\gtta \dd\mu \, .
\end{equation}
The two Dirichlet forms for the disordered and non-disordered case 
are respectively
\begin{equation}
\label{dirichlet disorder}
\cE_\mu(u) \, =\, -\crosqmu{Au}{u}\, ,
\end{equation}
and
\begin{equation}
\label{dirichlet non-disorder}
\cE(u)\, =\, -\crosq{ L_{q_0}u}{u} \, .
\end{equation}
As in \cite{cf:BGP}, we first prove a spectral gap type inequality that involves  the scalar product
$\cro{\cdot}{\cdot}_{2,q_0}$. For this we introduce  the projections on the line spanned  by $q^\prime_0$ in the spaces $L^2_{{q_0},\mu}$ and
$L^2_{{q_0}}$
\begin{equation}
 P_{2,{q_0},\mu} u=\frac{\croLqmu{u}{q^\prime_0}}{\croLq{q^\prime_0}{q^\prime_0}}q^\prime_0 \qquad \text{for all}
\,u=u(\gtta,\go)\in L^2_{{q_0},\mu} \, ,
\end{equation}
and
\begin{equation}
 P_{2,{q_0}}u\, =\, \frac{\croLq{u}{q^\prime_0}}{\croLq{q^\prime_0}{q^\prime_0}}q^\prime_0 \qquad \text{for all}\ u
 \in
L^2_{{q_0}} \, .
\end{equation}
Remark that since $q^\prime_0$ does not depend on $\go$,
\begin{equation}
 \croLqmu{q^\prime_0}{q^\prime_0}\, =\, \croLq{q^\prime_0}{q^\prime_0} \quad 
 \text{ and }  \quad \crosqmu{q^\prime_0}{q^\prime_0}\, =\, \crosq{q^\prime_0}{q^\prime_0}\, ,
\end{equation}
and that for all $u\in L^2_{{q_0},\mu}$
\begin{equation}
\label{link projections}
P_{2,{q_0},\mu}u\, =\, \langle P_{2,{q_0}}u \rangle_\mu\, =\, P_{2,{q_0}}\langle u\rangle_\mu\, . 
\end{equation}

\begin{proposition}
\label{prop:min dirichlet des}
 For all $u\in L^2_{{q_0},\mu}$ such that for almost every $\go$, $\intS u(\cdot, \go) =0$  
\begin{equation}
\cE_{\mu}(u)\, \geq \, c_K \croLqmu{u-P_{2,{q_0},\mu}u}{u-P_{2,{q_0},\mu}u}\, ,
\end{equation}
with
\begin{equation}
c_K\, =\, 1-K(1-r_0^2)\in(0,1/2)\, .
\end{equation}
\end{proposition}

The proof of this proposition relies on the corresponding result for the non-disordered case.:

\begin{proposition}
(see  
\cite[Prop. 2.3]{cf:BGP})
\label{prop:min dirichlet non-des}
 For all $u\in L^2_{{q_0}}$ such that for almost every $\go$, $\intS u(\cdot, \go) =0$
\begin{equation}
\cE(v)\, \geq \, c_K\croLq{ u-P_{2,{q_0}}u}{u-P_{2,{q_0}}u}\, . 
\end{equation}
\end{proposition}
\medskip

\noindent
\textit{Proof of Proposition \eqref{prop:min dirichlet des}.}
The first step of the proof is to make the Dirichlet form of the non-disordered case appear in the the disordered case one, that is
\begin{align}
 \cE_\mu(u) & \, =\, \langle \cE(u) \rangle_\mu +\int_{\bbR}\int_{(\bbS )^2} u(\theta,\go)\tilde J*[u(\theta,\go)-u(\theta,\go')]\dd\theta
\dd\mu\otimes \mu\\ \label{link Dirichlet}
        & \, =\, \langle \cE(u) \rangle_\mu +\frac 12 \int_{\bbR}\int_{(\bbS )^2} [u(\theta,\go)-u(\theta,\go')]\tilde
J*[u(\theta,\go)-u(\theta,\go')]\dd \theta \dd\mu\otimes \mu\, ,
\end{align}
and from Proposition \eqref{prop:min dirichlet non-des} we see  that
\begin{equation}
 \langle \cE(u)\rangle_\mu \, \geq \, c_K \croLq{ u-P_{2,{q_0}}u}{u-P_{2,{q_0}}u}\, .
\end{equation}
Now remark that if we define
\begin{equation}
 v\, =\, u-P_{2,{q_0},\mu}u\, ,
\end{equation}
using \eqref{link projections} we get
\begin{equation}
 v-P_{2,{q_0}}v\, =\, u-P_{2,{q_0}}u\, ,
\end{equation}
and so
\begin{equation}
 \langle \cE(u)\rangle_\mu \, \geq \, c_K \croLqmu{v-P_{2,{q_0}}v}{v-P_{2,{q_0}}v}\, .
\end{equation}

We now introduce an orthogonal decomposition of the space $L^2_{{q_0}}$ which is well adapted to the convolution with $\tilde J$.

\begin{lemma}
(See \cite[Lemma 2.1]{cf:BGP}.)
\label{lem:decomposition L2q}
 We have the following decomposition
\begin{equation}
L^2_{{q_0}}\, =\, F_0\oplus^\bot F_{1/2}\oplus^\bot F_{K-1/2}
\end{equation}
where
\begin{equation}
 F_0\, :=\, \left\{\theta\mapsto a_0+\sum_{j\geq 2}a_j\cos(j\theta)+b_j\sin(j\theta)\, ;\, \sum_j a_j^2+b^2_j<\infty \right\}
\end{equation}
and both $F_{1/2}$ and $F_{K-1/2}$ are one dimensional subspaces generated respectively by $\theta\mapsto \sin(\theta)q(\theta)\,
(=-q^\prime_0(\theta)/2Kr_0)$ and by $\theta \mapsto \cos(\theta) q_0 (\theta)$. Moreover, when $u\in F_\lambda$, then 
\begin{equation}
 \tilde J *u\, =\, \frac{\lambda}{q_0} u\, .
\end{equation}
\end{lemma}
\medskip

With the help of Lemma~\ref{lem:decomposition L2q} we can find a lower bound for the last term in \eqref{link Dirichlet}: choose $\ga$ such that $P_{2,{q_0}}u=\alpha q^\prime_0$, so that
we can write
\begin{equation}
\label{min Dirchlet}
  \cE_\mu(u) \, \geq \,c_K \croLqmu{ v-P_{2,{q_0}}v}{v-P_{2,{q_0}}v} + \frac{\croLq{q^\prime_0}{q^\prime_0}}{4} \int_{(\bbS )^2}
(\alpha(\go)-\alpha(\go'))^2 \dd\mu\otimes \mu\,  .
\end{equation}
But if  $P_{2,{q_0}}v=\beta q^\prime_0$ (recall that $v=u-P_{2,{q_0},\mu}u$), then since $P_{2,{q_0},\mu}u$ is colinear to
$q^\prime_0$,  for almost all $\go$, $\go'$
\begin{equation}
 \beta(\go)-\beta(\go')\, =\, \alpha(\go)-\alpha(\go')
\end{equation}
and since $v$ is orthogonal to $q^\prime_0$ (with respect to $\langle \cdot,\cdot\rangle_{2,q_0,\mu}$) we get
\begin{equation}
 \int_{\bbR} \beta(\go)\dd\mu\, =\, 0 \, .
\end{equation}
So \eqref{min Dirchlet} becomes
\begin{equation}
  \cE_\mu(u) \, \geq \,c_K \croLqmu{v-P_{2,{q_0}}v}{v-P_{2,{q_0}}v} + \frac{\croLq{q^\prime_0}{q^\prime_0}}{2} \int_{\bbS } \beta^2(\go) \dd\mu\,
.
\end{equation}
It is sufficient to compare this last minoration with the norm $\croLqmu{v}{v}$, and from Lemma \ref{lem:decomposition L2q} it comes
\begin{equation}
 \croLqmu{v}{v}\, =\,  \croLqmu{ v-P_{2,{q_0}}v}{v-P_{2,{q_0}}v} + \croLq{ q^\prime_0}{q^\prime_0} \int_{\bbS } \beta^2(\go) \dd\mu\, .
\end{equation}
This completes the proof of Proposition~\ref{prop:min dirichlet des}.
\qed

We now need two lemmas comparing the scalar products $\croLqmu{\cdot}{\cdot} $ and $\crosqmu{\cdot}{\cdot}  $. They correspond to
Lemmas 2.4 and 2.5 in \cite{cf:BGP}. Their proofs are very similar to the proofs
of the results corresponding results in \cite{cf:BGP}
(to which we refer also for the explicit values of the constants $C$ and $c$ appearing below)
and they use  in particular the rigged
Hilbert space representation of $H^{-1}_{{q_0},\mu}$ (see \cite[p.82]{MR697382}): namely, one can identify $H^{-1}_{{q_0},\mu}$ as
the dual space $V'$ of the space $V$ closure of $\cD(A)$ with respect to the norm $\Vert{u}\Vert_V:= \left( \int_{\bbR\times\bbS} v'(\gtta,
\go)^2
\dd\gtta\mu(\dd\go) \right)^\frac12$. The pivot space $H$ is the usual
$L^2(\gl\otimes\mu)$ (endowed with the Hilbert norm $\Vert u\Vert_{2,\mu}:=\left( \intSR u(\gtta,
\go)^2\dd\gtta\mu(\dd\go) \right)^\frac12$). In particular, one easily sees that the inclusion $V\subseteq H$ is dense.
Consequently, one can define $T:H\rightarrow V'$ by setting $Tu(v)= \intSR u(\gtta, \go) v(\gtta, \go)\dd\gtta \mu(\dd\go)$.
One can prove that $T$ continuously injects $H$ into $V'$ and that $T(H)$ is dense into $V'$ so that one can identify $u\in H$
with $Tu\in V'$. Then for $u\in H$,
\begin{equation}
 \Vert u\Vert_{V'} = \Vert Tu\Vert_{V'} = \sup_{v\in V} \frac{\int \cU v'}{\Vert v\Vert_{V}}= \sqrt{\int \frac{\cU^2}{q_0}}\, ,
\end{equation}
which enables us to identify $H^{-1}_{{q_0},\mu}$ with $V'$.

\medskip

We define the projection in
$H^{-1}_{{q_0},\mu}$:
\begin{equation}
P_{-1,{q_0},\mu}u\, =\, \frac{\crosqmu{u}{q^\prime_0}}{\crosq{q^\prime_0}{q^\prime_0}}q^\prime_0\, .
\end{equation}
\begin{lemma}
 \label{lem:comparaison norm projections}
 For every $K>1$ there exists a constant $C=C(K)>0$ such that
for $u\in L^2_\mu$ such that $\intS u =0$  for almost every $\go$
\begin{equation}
\begin{split}
\croLqmu{u-P_{2,{q_0},\mu}u}{u-P_{2,{q_0},\mu}u}\, &\geq \, e^{4Kr_0}C
\croLqmu{u-P_{-1,{q_0},\mu}u}{u-P_{-1,{q_0},\mu}u}\\ 
&\geq \, C \crosqmu{u-P_{-1,{q_0},\mu}u}{u-p_{-1,{q_0},\mu}u}\, .
\end{split}
\end{equation}
\end{lemma}
\begin{lemma}
\label{lem:comparaison norm}
For every $K>1$ there exists $c=c(K)>0$ such that for $u\in L^2_\mu$ such that $\intS u =0$  for almost every $\go$ and 
\begin{equation}
 \crosqmu{u}{u}\, \geq \, c \croLqmu{P_{2,{q_0},\mu}u}{P_{2,{q_0},\mu}u}\, .
\end{equation}
\end{lemma}

\medskip
\noindent
{\it Proof of Proposition~\ref{th:spectral gap A}.}
Of course Proposition \ref{prop:min dirichlet des} and Lemma \ref{lem:comparaison norm projections} imply directly the spectral
gap inequality for the Dirichlet form:
\begin{equation}
\label{low bound Dirichlet}
\cE(u) \, \geq \, c_KC  \crosqmu{u-P_{-1,{q_0},\mu}u}{u-P_{-1,{q_0},\mu}u}\qquad \text{for all}\, u\in H^{-1}_{{q_0},\mu}\, .
\end{equation}
We now prove the self-adjoint property of $A$. It is sufficient to prove that the range of $1-A$ is dense in $H^{-1}_{\mu}$ (see
\cite[p.113]{MR697382}). 
For $u, v \in D(A)$, we have
\begin{multline}
\label{expr scal prod 1-A}
 \crosqmu{v}{(1-A)u}\, =\, -\int_\bbR \intS v(\gtta,\go)\left(\int_0^\gtta \frac{\cU}{q_0} \right)\dd\gtta\dd\mu+\frac12 \int_\bbR\intS \frac{vu}{q_0}\dd\gtta\dd\mu \\
-\int_{\bbR}\int_{(\bbS )^2} v(\theta,\go)\tilde J*u(\theta,\go')\dd\theta \dd\mu \otimes\mu \, .
\end{multline}
The right side of this expression is still defined for $u,v\in L^2(\gl\otimes\mu)$ (recall that $\gl$ denotes the Lebesgue measure
on $\bbS $, and that we denote the usual scalar product on $L^2(\gl\otimes\mu)$ by $\Vert\cdot\Vert_{2,\mu}$) and there exists
$c>0$ such that 
\begin{equation}
\label{ineq:1-A norm vs L2}
 \crosqmu{v}{(1-A)u}\, \leq\, c \Vert u\Vert_{2,\mu}  \Vert v\Vert_{2,\mu}\, ,
\end{equation}
Furthermore from \eqref{low bound Dirichlet} and Lemma \ref{lem:comparaison norm} we have
\begin{equation}
\label{ineq:1-A norm vs L2 bis}
 \crosqmu{u}{(1-A)u}\, \geq\, \frac1c \Vert u\Vert_{2,\mu}^2\, .
\end{equation}
So the bilinear form $(u,v)\mapsto  \crosqmu{v}{(1-A)u}$ is continuous and coercive on $H^{-1}_{\mu}\times H^{-1}_{\mu}$. If
$f\in H^{-1}_{\mu}$, the linear form $v\mapsto \crosqmu{v}{f} $ is continuous on $L^2(\gl\otimes\mu)$, therefore from Lax-Milgram
Theorem we get that there exists a unique $u\in L^2(\gl\otimes\mu)$ such that for all $v\in L^2(\gl\otimes\mu)$
\begin{equation}
 \crosqmu{v}{(1-A)u} \, =\, \crosqmu{v}{f} \, .
\end{equation}
Since
\begin{equation}
 \crosqmu{v}{f} \, =\, -\int_\bbR \intS v(\gtta,\go)\left(\int_0^\gtta \frac{\cF}{q_0} \right)\dd\gtta\dd\mu\, ,
\end{equation}
from \eqref{expr scal prod 1-A} we obtain that for almost $\gtta$ and $\go$
\begin{equation}
 -\int_0^\gtta \frac{\cU(\gtta',\go)}{q_0(\gtta')}\dd\gtta' +\frac{u(\gtta,\go)}{2q_0(\gtta)} -\int_\bbR \left(\tilde
J*u\right)(\gtta,\go)\dd\mu\, =\, -\int_0^\gtta\frac{\cF(\gtta',\go)}{q_0(\gtta')}\dd\gtta'\, .
\end{equation}
So it is clear that if $f$ is continuous with respect to $\gtta$, then $u$ has a version $C^2$ with respect to $\gtta$. Thus $u\in
D(A)$ and applying $\partial_\gtta(q_0(\gtta)\partial_\gtta \cdot)$ to the both sides of this last expression, we get $(1-A)u=f$.
Since this kind of functions $f$ is dense in $H^{-1}_{\mu}$, we can conclude that the range of $1-A$ is dense, and that $A$ is
essentially self-adjoint.
This completes the proof of Proposition~\ref{th:spectral gap A}.
\qed

\section{Perturbation arguments (completion of the proof of Theorem~\ref{th:expansion speed})}

\label{sec:pa}
In this section we complete the proof of
Theorem \ref{th:expansion speed}. 
Essentially, this section is devoted to computing
 the expansion of the speed $c_\mu(\gd)$ in. We first recall a lemma that gives a useful parametrization
in the neighborhood of $M_0$. The proof of this lemma is given in \cite{cf:SellYou} , and it is used in the proof of Theorem \ref{th:M} (see \cite{cf:GPPP,cf:SellYou}).

\begin{lemma}
\label{lem:parametrisation}
There exists a $\gs>0$ such that for all $p$ in the neighborhood 
\begin{equation}
\label{eq:Nsigma}
N_\gs\, :=\, 
\cup_{q\in M_0} B_{L^2(\gl\otimes\mu)}(q,\sigma)\, ,
\end{equation} 
 of $M_0$ there is one
and only one $q=v(p) \in M_0$ such that $\crosqmu{p-q}{\partial_\gtta q}=0$. Furthermore the mapping $p \mapsto v(p)$ is in
$C^\infty(X^1_\mu,X^1_\mu)$, and
\begin{equation}
 Dv(p)\, =\, P^o_{v(p)}\, .
\end{equation}
\end{lemma}

\medskip

\noindent
{\it Proof of Theorem \ref{th:expansion speed}.}
The existence and stability
of a {\sl rotating} solution $\tilde q_\gd(\gtta-\psi-c_\mu(\gd)t)$ of \eqref{FKP kuramoto disorder delta} ($\psi$ is arbitrary)
 has been established in Section~\ref{sec:hyperbolic structure} for  $\gd\leq \gd_0$, 
see Theorem~\ref{th:M} and the two remarks that follow it.
We are left with proving Theorem \ref{th:expansion speed}(2).

Thanks to the invariance by rotation, we can define $\tilde q_\gd$ such that $v(\tilde q_\gd)=q_0 $.
Now if we denote 
\begin{equation}
 n_\gd\, :=\, \tilde q_\gd-v\left(\tilde q_\gd\right)\, ,
\end{equation}
then $n_\gd$ verifies $n_\gd=\phi_\gd(q_0)$ and (see Lemma \ref{lem:parametrisation})
\begin{equation}
\label{eq:orth n q}
 \crosqmu{n_\gd}{q^\prime_0}=0
\end{equation}
\begin{equation}
\label{eq:orth An q}
 \crosqmu{An_\gd}{q^\prime_0}=0\, .
\end{equation}
Moreover the estimates we have on the mapping $\phi_\gd$ in Theorem \ref{th:M} give
\begin{equation}
\label{ineq:bound n}
 \Vert n_\gd \Vert_{2,\mu}\, \leq\, C\gd\, ,
\end{equation}
\begin{equation}
\label{ineq:bound dn}
 \Vert \partial_\gtta n_\gd \Vert_{2,\mu}\, \leq C \gd\, .
\end{equation}
Taking the derivative with respect to $t$, at time $t=0$, we get (we recall the notation $p^{(\psi)}_t(\gtta,\go)=\tilde q_\gd(\gtta-\psi-c_\mu(\gd)t)$) :
\begin{equation}
\label{eq:equality derivate}
 -c_\mu(\gd)\left(q^\prime_0+\partial_\gtta n_\gd\right)\, =\, \partial_t p^{(0)}_0\, .
\end{equation}
So \eqref{FKP kuramoto disorder delta} at time $t=0$ becomes (recall that $q_0$ is a stationary solution of \eqref{FKP kuramoto disorder without drift}) :
\begin{equation}
\label{eq:FKP delta projected}
 -c_\mu(\gd)\left( q^\prime_0+\partial_\gtta n_\gd\right)\, =\, A n_\gd -\partial_\gtta\left[n_\gd\langle J*n_\gd\rangle_\mu\right]-\gd\go q^\prime_0-\gd\go\partial_\gtta n_\gd\, .
\end{equation}
From \eqref{ineq:bound n} we deduce the bound
\begin{equation}
\label{ineq:bound nJn}
\left\Vert \partial_\gtta\left[n_\gd\langle J*n_\gd\rangle_\mu\right]\right\Vert_{-1,\mu}\, \leq \, \Vert J\Vert_2 C^2
\gd^2\, ,
\end{equation}
so by taking the $H^{-1}_{{q_0},\mu}$ scalar product of $q^\prime$ in \eqref{eq:FKP delta projected}, using \eqref{eq:orth An q},
\eqref{ineq:bound n}, \eqref{ineq:bound dn} and the fact that $\int_\bbR w\dd\mu=0$, we get that $c_\mu(\gd)$ is of order $\gd^2$. This implies,
using the same arguments, that
\begin{equation}
 \left\Vert A n_\gd-\gd\go q^\prime_0 \right\Vert_{-1,\mu}\, =\, O(\gd^2)\, .
\end{equation}
So 
\begin{equation}
\label{comparaison n0 n1}
 \Vert A(n_\gd-\gd n^{(1)}) \Vert_{-1,\mu}\, =\, O(\gd^2)\, ,
\end{equation}
and since $\Vert (1-A)^{(1/2)}u\Vert_{-1,\mu} \sim \Vert u \Vert_{2,\mu}$ (see \eqref{ineq:1-A norm vs L2} and \eqref{ineq:1-A norm vs L2 bis}), we have in particular
\begin{equation}
 \Vert n_\gd-\gd n^{(1)} \Vert_{2,\mu}\, =\, O(\gd^2)\, .
\end{equation}
It allows us to make a second order expansion for $c_\mu(\gd)$ : taking again the $H^{-1}_{{q_0},\mu}$ scalar product of $q^\prime_0$ in
\eqref{eq:FKP delta projected}, using the same bounds as for the first order expansion and \eqref{comparaison n0 n1}, we get :
\begin{equation}
\label{eq:second order expansion psi prime}
 c_\mu(\gd)\, =\, \gd^2 \frac{\crosqmu{\go\partial_\gtta n^{(1)}+n^{(1)}\langle J*n^{(1)}\rangle_\mu}{q^\prime_0}}{\crosqmu{q^\prime_0}{q^\prime_0}} +O(\gd^3)\, .
\end{equation}
Indeed, from \eqref{comparaison n0 n1}, $\Vert \go\partial_\gtta(n_\gd-\gd n^{(1)}) \Vert_{-1,\mu}$, $\Vert
\partial_\gtta[(n_\gd-\gd n^{(1)})\langle J*n^{(1)}\rangle_\mu]\Vert_{-1,\mu}$, $ \Vert \partial_\gtta[n^{(1)}\langle J*(n_\gd-\gd
n^{(1)})\rangle_\mu]\Vert_{-1,\mu}$ are of order $\gd^2$ and $ \Vert \partial_\gtta[(n_\gd-\gd n^{(1)})\langle J*(n_\gd-\gd
n^{(1)})\rangle_\mu]\Vert_{-1,\mu}$ of order $\gd^4$.
Since $c_\mu(\gd)$ is odd with respect to $\gd$, the second order term in \eqref{eq:second order expansion psi prime} is equal
to $0$. It is possible to get this fact directly : we remark that $n^{(1)}$ satisfies :
\begin{equation}
\label{eq:kdge3}
 L_{q_0}\int_\bbR n^{(1)}\dd\mu \, =\, \int_\bbR A n^{(1)} \dd\mu\, =\, \left(\int_\bbR \go\dd\mu\right)q^\prime_0\, =\, 0\, ,
\end{equation}
\begin{equation}
 \crosq{\int_\bbR n^{(1)}\dd\mu}{q^\prime_0}\, =\, \crosqmu{n^{(1)}}{q^\prime_0}\, =\, 0\, .
\end{equation}
So since $L_{q_0}$ is bijective on the orthogonal of $q^\prime_0$ in $H^{-1}_{1/q}$ (see \cite{cf:BGP}), we have $\int_\bbR n^{(1)} \dd\mu=0$ and
$\langle J*n^{(1)}\rangle\mu=0$. On the other hand, since the operator $A$ conserves the parity with respect to $\gtta$, $n^{(1)}$ is odd with respect to $\gtta$ and thus
\begin{equation}
 \crosqmu{\go\partial_\gtta n^{(1)}}{q^\prime_0}\, =\, \int_\bbS \int_\bbR \frac{\go n^{(1)}}{q_0}\left(q_0-\frac{1}{2\pi I_0^2(2Kr_0)}\right)\dd\gtta\dd\mu\, =\, 0\, .
\end{equation}
Now back to \eqref{eq:FKP delta projected}: since $c_\mu(\gd)$ is of order $\gd^3$ and using $\int_\bbS n^{(1)}\dd\mu=0$, we get
\begin{equation}
 \left\Vert A \left(n_\gd-\gd n^{(1)}-\gd^2\go\partial_\gtta n^{(1)} \right)\right\Vert_{-1,\mu}\, =\, O(\gd^3)\, ,
\end{equation}
and thus
\begin{equation}
\label{eq:n_2}
 \Vert n_\gd-\gd n^{(1)}-\gd^2 n^{(2)}\Vert_{2,\mu} \, =\, O(\gd^3)\, .
\end{equation}
This allows us this time to do a third order expansion in \eqref{eq:FKP delta projected} :
\begin{equation}
\label{eq:rt5p}
 c_\mu(\gd)\, =\, \gd^3\frac{\crosqmu{\go\partial_\gtta n^{(2)}}{q^\prime_0}}{\crosqmu{q^\prime_0}{q^\prime_0}} + O(\gd^4)\, .
\end{equation}
This procedure may be repeated recursively at any order: we do not go through the 
details again, but we do report the result below (Remark~\ref{rem:pa}) and  
 we  
point out  that
the $O(\gd^4)$ \eqref{eq:rt5p} turns out to be  $O(\gd^5)$, 
in agreement with the fact that  $c_\mu(\gd)$ is odd in $\gd$.
\qed

\medskip

\begin{rem}
\label{rem:pa}
\rm
As anticipated above, one can get arbitrarily many terms
in the formal series $c_\mu(\gd)= \sum_{i=1,2, \ldots} c_{2i+1} \gd^{2i+1}$
and the remainder, when the series is stopped at $i=n$, is $O(\gd^{2i+3})$.
In fact, by arguing like above,  we have
\begin{equation}
 c_5\, =\, \frac{\left\langle \partial_\gtta [n^{(2)}\langle J*n^{(3)}\rangle_\mu]+ \partial_\gtta [n^{(3)}\langle J*n^{(2)}\rangle_\mu]+ w\partial_\gtta n^{(4)} ,q'_0\right\rangle_{-1,q_0,\mu}}{\langle q'_0,q'_0 \rangle_{-1,q_0}}\, ,
\end{equation}
where
\begin{equation}
 An^{(3)}\, =\,  \partial_\gtta [n^{(1)}\langle J*n^{(2)}\rangle_\mu]+ w\partial_\gtta n^{(2)}-\frac{\left\langle  w\partial_\gtta n^{(2)} ,q'_0\right\rangle_{-1,q_0,\mu}}{\langle q'_0,q'_0 \rangle_{-1,q_0}}q'_0 \, , 
\end{equation}
and
\begin{equation}
 An^{(4)}\, =\, \partial_\gtta [n^{(2)}\langle J*n^{(2)}\rangle_\mu]+ \partial_\gtta [n^{(1)}\langle J*n^{(3)}\rangle_\mu]+ w\partial_\gtta n^{(3)}\, .
\end{equation}
Actually, 
by induction we obtain 
\begin{equation}
 c_{2i+1}\, =\, \frac{\left\langle \sum_{k+l=2i+1,k>0,l>0}\partial_\gtta [n^{(l)}\langle J*n^{(k)}\rangle_\mu]+  w\partial_\gtta n^{(2i)} ,q'_0\right\rangle_{-1,q_0,\mu}}{\langle q'_0,q'_0 \rangle_{-1,q_0}}\, ,
\end{equation}
and
\begin{equation}
 n^{(2i)}\, =\, \sum_{k+l=2i,k>0,l>0}\partial_\gtta [n^{(l)}\langle J*n^{(k)}\rangle_\mu] + w\partial_\gtta n^{(2i-1)}\, ,
\end{equation}
\begin{equation}
 n^{(2i+1)}\, =\, \sum_{k+l=2i+1,k>0,l>0}\partial_\gtta [n^{(l)}\langle J*n^{(k)}\rangle_\mu] + w\partial_\gtta n^{(2i)}- c_{2i+1}q'_0\, .
\end{equation}
Since this procedure yields also $n^{(j)}$ for arbitrary $j$, one
can generalizes also \eqref{eq:n_2} and, hence, \eqref{eq:qgddevel}. 
\end{rem}

\section{Active rotators}
\label{sec:AR}

In this section we  deal with the equation \eqref{eq:AR} and we do it in a rather informal  way, because on one hand a formal statement would be very close to
Theorem~\ref{th:expansion speed} and, on the the other hand, the large scale behavior 
of disordered active rotators is qualitatively and quantitatively close to the non disordered case, treated in
 \cite{cf:GPPP}, in a way that we explain below. 

First of all, from a technical viewpoint the main difference between
\eqref{eq:AR} and \eqref{FKP kuramoto disorder} is that \eqref{eq:AR}
  is (in general) not rotation invariant, so the manifold $M_\gd=\{q_\psi+\phi(q_\psi)\}$ we get after perturbation is not necessarily a circle. Unlike Theorem \ref{th:expansion speed}, the motion on $M_\gd$ is not uniform, and we describe the behaviour on $M_\gd$ by the phase derivate $\dot{\psi}$. We follow the same procedure as in the previous section : if $p^\gd_t$ is a solution \eqref{eq:AR} belonging to $M_\gd$, we define (see Lemma \ref{lem:parametrisation})
\begin{equation}
 q_{\psi^\gd_t}\, =\, v(p^\gd_t)\ ,  \ \ \text{ and } \ \ \
 n^\gd_t\,=\,  p^\gd_t- v(p^\gd_t)\, .
\end{equation}
In this context, \eqref{eq:FKP delta projected} becomes
\begin{equation}
\label{eq:FKP delta projected AR}
 -\dot{\psi}^\gd_t q'_{\psi^\gd_t} + \partial_t n^\gd_t\, =\, A^{\psi^\gd_t} n^\gd_t -\partial_\theta [n^\gd_t \langle J*n^\gd_t \rangle_\mu ]-\gd U q'_\psi  -\gd U\partial_\theta n^\gd_t \, ,
\end{equation}
where $A^\psi$ is the rotation of the operator $A$
\begin{equation}
 A^\psi u(\gtta,\go) \, :=\, \frac12 \Delta u(\gtta,\go)-\partial_\gtta\Big(q_0(\gtta-\psi)\langle J*u\rangle_\mu(\gtta)+u(\gtta,\go)J*q_0(\gtta-\psi)\Big)\, .
\end{equation}
Note that we can reformulate the second term of the left hand side in \eqref{eq:FKP delta projected AR}:
\begin{equation}
 \partial_t n^\gd_t\, =\, \dot{\psi}^\gd_t\partial_\psi \phi(q_\psi) |_{\psi=\psi^\gd_t}\, .
\end{equation}
So, as in the previous section, using the estimates on the mapping $\phi$ given in Theorem \ref{th:M}, we get the bounds
\begin{equation}
   \Vert n^\gd_t \Vert_{2,\mu}\, \leq\, C\gd \, , \ \
  \Vert \partial_t n^\gd_t \Vert_{2,\mu}\, \leq\, C\gd  |\dot{\psi}^\gd_t|
\ \text{ and } \ 
   \Vert \partial_\theta [n^\gd_t \langle J*n^\gd_t \rangle_\mu ] \Vert_{2,\mu} \, \leq\, \Vert J\Vert_2 C^2\gd \, ,
\end{equation}
and we deduce the first order expansion 
\begin{equation}
\label{eq:AReq}
 {\dot{\psi}}^\gd_t\, 
 =\, \gd\frac{ \langle (U q_{\psi^\gd_t})', q'_{\psi^\gd_t} \rangle_{-1,q_{\psi^\gd_t},\mu}}{ \langle q'_0,q'_0 \rangle_{-1,q_0}}+O(\gd^2)\, .
\end{equation}
Since $\dot{\psi}$ is odd in $\gd$ and the expansion can be pushed further in $\gd$, this $O(\gd^2)$ is in reality a $O(\gd^3)$ and one can actually improve this result both in the direction
of obtaining a regularity estimate on the $O(\gd^2)$ rest in \eqref{eq:AReq}
(like in \cite[Th.~2.3]{cf:GPPP}) and of going to higher orders (like in Remark~\ref{rem:pa}). 

However the evolution for small $\gd$ is dominated by the leading order and from
\eqref{eq:AReq} we can directly read that, to first order, the effect of the disorder
is rather simple: in fact
\begin{equation}
 \langle (U q_{\psi})', q'_{\psi} \rangle_{-1,q_{\psi},\mu}\, =\, 
 \int_\bbR \int_\bbS U(\theta, \go) q_{\psi} (\theta)\left( q_{\psi}(\theta)- c\right) \dd
  \theta \mu (\dd \go)\, ,
\end{equation}
where $c$ is such that $\int_\bbS (q_\psi -c)=0$, that is $1/c={2\pi (I_0(2Kr_0))^2}$
(recall \eqref{eq:defstationarysolution nodisorder}-\eqref{def bessel}: this computation
is analogous to \eqref{eq:kdge3}).
Since the integrand depends on $\go$ only via $U$, this integration can be performed first 
and the system behaves to leading order in $\gd$ as the non-disordered model
with active rotator dynamics led by the deterministic force
$\int_\bbR U(\cdot, \go) \mu(\dd \go)$. The rich phenomenology connected 
to these models is worked out in \cite[Sec.~3]{cf:GPPP}.

\section{Symmetric case: stability of the stationary solutions}
\label{sec:sym}

\subsection{On the  non-trivial stationary solutions (proof of Lemma \ref{th: Psimu concave})}
We start by observing that
in the case with no disorder  the strict concavity of the fixed-point function $\Psi_0$ has been proven in \cite[Lemma 4, p.315]{cf:Pearce}, in
the apparently different context of classical XY-spin model (for a detailed discussion on the link with these models see \cite{cf:BGP}). We are
going to obtain the concavity of $\Psi^\mu_\gd$ for small $\gd$ via a perturbation argument, by relying on the result  in \cite{cf:Pearce}.

%


Since $\Psi^\mu_\gd$ is a smooth perturbation of $\Psi_0$, one expects that the strict concavity of $\Psi_0$ will be preserved
to $\Psi^\mu_\gd$ for small $\gd>0$, namely $\sup_{x} (\Psi_\gd^\mu)''(x)<0$. Nevertheless, an easy calculation shows that
$\Psi''_0(0)=0$; in that sense one has to treat the concavity in a neighborhood of $0$ as a special case. 

\medskip

\noindent
In what follows, we suppose that the coupling strength $K$ is bounded above and below by fixed constants $K_{\min}$ and $K_{\max}$:
\begin{equation}0\, <\, K_{\min} \, \leq\,  K\, \leq\, K_{\max}\, <\, \infty\, .\end{equation}

We first prove the statement on the concavity in a neighborhood of $0$: there exist $\eta_0>0$, $\gd>0$ such
that for all $K\in[K_{\min}, K_{\max}]$, for all $\mu$ such that $\Supp(\mu)\subseteq[-1,1]$, $\Psi_\gd^\mu$ is strictly concave on $[0,
\eta_0]$.
\medskip

Indeed, one easily shows (using that the function $x\mapsto \Psi_\mu^\gd(x)$ is odd) that we have the following Taylor's expansion:
\begin{equation}(\Psi_\gd^\mu)''(x) = -6 D^\gd(\mu)K^3 x + \epsilon(x)\, ,\end{equation}where $\epsilon(x)=o(x)$ as $x\rightarrow 0$ and where
for fixed $\mu$, we write \begin{equation}D^\gd(\mu) := \int_{\R{}}{h(\gd\go)\mu(\dd\go)},\end{equation}where \begin{equation}h(\go)
:= \frac{1}{2(1+\go^{2})}-\frac{8\go^{2}}{(1+4\go^{2})^2}.\end{equation}
Note that the $o(x)$ only depends on $K_{\max}$ (in particular it can be chosen independently of $\mu$).
A closer look at the function $h$ shows that there exists $\gd>0$ such that for all $\mu$ with $\Supp(\mu)\subseteq[-1, 1]$, $D^\gd(\mu)>\frac
14$.
If we choose $\eta_{0}>0$ such that $\frac{1}{\eta_0}\sup\limits_{0\leq x< \eta_0}|\epsilon(x)|< \frac 32 K_{\min}^3$ then
$(\Psi_\gd^\mu)''(x)<0$ for all $0<x<\eta_0$, which is the desired result.

\medskip

We are now left with proving concavity away from $0$: namely, we prove that for all $\eta>0$, all $K_{\max}$, there exists $\gd_0>0$ such that
for all $K\leq K_{\max}$, for all $0<\gd<\gd_0$, for any measure $\mu$ such that $\Supp(\mu)\subseteq[-1, 1]$, $\Psi^\mu_\gd$ is strictly concave
on $[\eta, 2K_{\max}]$.
\medskip

Indeed, using the strict concavity of $\Psi_0$ proved in \cite{cf:Pearce}, there exists a constant $\alpha>0$ such that for all
$x\in[\eta, 2K_{\max}]$, $\Psi_{0}''(x)<-\alpha<0$. But then, it easy to see that \begin{equation}\sup_{0<\gd<\gd_0}\sup_{\mu,\
\Supp(\mu)\subseteq [-1,
1]}\sup_{x\in[0, 2K_{\max}]} \left|(\Psi^\mu_\gd)''(x) - \Psi_{0}''(x)\right|
\stackrel{\gd_0\searrow 0}{\rightarrow}0.\end{equation}If one chooses
$\gd_0$ such that the latter quantity is smaller than or equal to $\frac{\alpha}{2}$, the result follows. The proof of Lemma \ref{th: Psimu concave} is therefore complete.
\qed

\subsection{On the linear stability of non-trivial stationary solutions}
We now prove Theorem \ref{th:spectral prop L disorder} along with a number of explicit estimates.

\begin{rem}\rm
 Note that, since the whole operator $L^\go_{{q}}$ is no longer self-adjoint nor symmetric, its spectrum need not be real. In that
extent, one has to deal in this section with the complexified versions of the scalar products defined in Section \ref{sec:mainresults},
\eqref{def scalar product sobolev with wage and mu} and in Section \ref{sec:hyperbolic structure}, \eqref{def scalar product L2 with wage and
mu}. Thus, we will assume for the rest of this section that we work with complex versions of these scalar products. The results
concerning the operator $A$ are obviously still valid, since $A$ is symmetric and real.

We will also use the following standard notations: for an operator $F$, we will denote by $\rho(F)$ the set of all complex numbers $\lambda$ for
which $\lambda-F$ is invertible, and  by $R(\lambda, F):= \left( \lambda - F \right)^{-1}$, $\lambda\in\rho(F)$ the resolvent of $F$.
The spectrum of $F$ will be denoted as $\sigma(F)$.
\end{rem}

\subsubsection{Decomposition of $L^\go_{{q}}$}
In what follows, $K>1$ and $r_0= \Psi_{0}(2Kr_0)>0$ are fixed.

In order to study the spectral properties of the operator $L^\go_{{q}}$ for general distribution of disorder, we decompose
$L^\go_{{q}}$ in
\eqref{eq:defLqmu} into the sum of the self-adjoint operator $A$ defined in \eqref{def A} and a perturbation $B$ which
will be considered to be small w.r.t.
$A$, namely:
%
\begin{equation}
 \label{eq:defBmu}
Bu(\gtta, \go)\, :=\,  - \partial_{\gtta} \left( u(\gtta, \go) \langle J \ast \gep(q) \rangle_\mu + \gep(q)(\gtta, \go, \gd)
\langle J \ast
u\rangle_\mu(\gtta) + \gd\go u(\gtta, \go)\right),
\end{equation}
where 
\begin{equation}
 \label{eq:defdeltaqmu}
\gep(q) \, :=\,  (\gtta, \go, \gd)\mapsto q(\gtta, \gd\go) - q_0(\gtta),
\end{equation}
is the difference between the stationary solution with disorder and the one without disorder.

 
\begin{proposition}
 \label{prop:Apositivemu}
The (extension of the) operator $A$ is the infinitesimal generator of a strongly continuous semi-group
of contractions $T_A(t)$ on $H^{-1}_{{q_0},\mu}$. 

Moreover, for every $0<\alpha<\frac{\pi}{2}$ this semigroup can be extended to an
analytic semigroup $T_A(z)$ defined on $\Delta_\alpha\, :=\, \ens{z\in\bbC}{|\arg(z)|<\alpha}$.
\end{proposition}

We recall here the result we use concerning analytic extensions of strongly continuous semigroups. Its proof can be found in
\cite[Th 5.2, p.61]{Pazy1983}.

\begin{proposition}
 \label{prop:pazysemgps}
Let $T(t)$ a uniformly bounded strongly continuous semigroup, whose infinitesimal generator $F$ is such that $0\in\rho(F)$ and let
$\alpha\in(0, \frac\pi2)$. The
following statements are equivalent:
\begin{enumerate}
 \item \label{it:prop:pazysemgps1}$T(t)$ can be extended to an analytic semigroup
in the sector $\Delta_\alpha\, =\, \ens{\lambda\in\bC}{|\arg(\lambda)|<\alpha}$ and $\Vert{T(z)}\Vert$ is uniformly bounded in every closed
sub-sector $\bar{\Delta}_\alpha'$, $\alpha'<\alpha$, of $\Delta_\alpha$,
\item \label{it:prop:pazysemgps2} There exists $M>0$ such that\begin{equation}\rho(F) \supset \Sigma\, =\,
\ens{\lambda
\in\bC}{|\arg(\lambda)|<\frac{\pi}{2}+\alpha} \cup \{0\},\end{equation}and\begin{equation}\Vert{R(\lambda, F)}\Vert \, \leq\,  \frac{M}{|\lambda|}, \quad
\lambda\in\Sigma, \lambda\neq0\, .\end{equation}
\end{enumerate}
\end{proposition}
%

\begin{proof}[Proof of Proposition \ref{prop:Apositivemu}]

The proof in Section \ref{sec:hyperbolic structure}, Theorem \ref{th:spectral gap A} of the self-adjointness of $A$ shows that $A$ satisfies
the hypothesis of Lumer-Phillips Theorem (see \cite[Th 4.3, p.14]{Pazy1983}): $A$ is the infinitesimal generator of a $C_0$ semi-group of
contractions denoted by $T_A(t)$.

The rest of the proof is devoted to show the existence of an analytic extension of this semigroup in a proper sector. We follow
here the lines of the proof of Th 5.2, p. 61-62, in \cite{Pazy1983}, but with explicit estimates on the resolvent, in order to quantify properly
the appropriate size of the perturbation.

Let us first replace the operator $A$ by a small perturbation: for all $\gep>0$, let $A_{\gep}\, :=\,  A-\gep$, so that
$0$ belongs to $\rho(A_\gep)$. The operator $A_{\gep}$ has the following
properties: as $A$, it generates a strongly continuous semigroup of operators (which is $T_{A, \gep}(t)=T_A(t)e^{-\gep t}$).

%
Since $A$ is self-adjoint, it is easy to see that 
\begin{equation}
\label{eq:estimRAeps1}
\forall
\lambda\in\bC\smallsetminus\bbR, \Nsqmu{R(\lambda, A_\gep)}\, \leq\, 
\frac{1}{|\Im(\lambda)|}\, , 
\end{equation}
and since the spectrum of $A$ is negative, for every $\lambda\in\bC$ such that $\Re(\lambda)>0$,
\begin{equation}
\label{eq:estimRAepsdef1}
\Nsqmu{R(\lambda, A_\gep)}\, \leq\,  \frac{1}{|\lambda|}\, .                                     
\end{equation}
For any $\alpha \in(0, \frac{\pi}{2})$, let 
\begin{equation}
 \Sigma_\alpha \,:=\, \ens{\lambda\in\bbC}{|\arg(\lambda)|<\frac{\pi}{2}+\alpha}\, .
\end{equation}
Let us prove that for $\lambda\in\Sigma_\alpha$,
\begin{equation}
\label{eq:estim R lambda alpha}
 \Nsqmu{R(\lambda, A_\gep)}\, \leq\, \frac{1}{1-\sin(\alpha)}\cdot \frac{1}{|\lambda|}\, .
\end{equation}
Note that \eqref{eq:estim R lambda alpha} is clear from \eqref{eq:estimRAeps1} and \eqref{eq:estimRAepsdef1} when $\lambda$ is such
that $\Re(\lambda)\geq0$.

Let us consider $\sigma>0, \tau\in\bbR$ to be chosen appropriately later.

Let us write the following Taylor expansion for
$R(\lambda, A_\gep)$ around $\sigma+i\tau$ (at least well defined in a neighborhood of $\sigma+i\tau$ since $\sigma>0$):
\begin{equation}
\label{eq:taylor resolvent}
 R(\lambda, A_\gep) \, =\,  \sum_{n=0}^{\infty}{R(\sigma+i\tau, A_\gep)^{n+1}((\sigma + i\tau)-\lambda)^n}\, .
\end{equation}

From now, we fix $\lambda\in\Sigma_\alpha$ with $\Re(\lambda)<0$.
This series $R(\lambda, A_\gep)$ is well defined in $\lambda$ if one can choose $\sigma$, $\tau$ and $k\in(0,1)$ such that $\Nsqmu{R(\sigma
+i\tau, A_\gep)}|\lambda-(\sigma+i\tau)|\leq k<1$. In particular, using \eqref{eq:estimRAeps1}, it suffices to have $|\lambda-(\sigma+i\tau)|\leq
k|\tau|$ and since $\sigma>0$ is
arbitrary, it suffices to find $k\in(0,1)$ and $\tau$ with $|\lambda-i\tau|\leq k|\tau|$ to obtain the convergence of \eqref{eq:taylor
resolvent}.
For this $\lambda\in\Sigma_\alpha$ with $\Re(\lambda)<0$, let us define $\lambda'$ and $\tau$ as in Figure~\ref{fig:angle alpha}. Then,
$|\lambda-i\tau|\leq |\lambda'-i\tau|= \sin(\alpha)|\tau|$ with $\sin(\alpha)\in(0,1)$. So
the series converges for $\lambda\in\Sigma_\alpha$ and one has, using again \eqref{eq:estimRAeps1},
\begin{equation}
 \label{eq:estimRAepsdef}
\Nsqmu{R(\lambda, A_\gep)} \, \leq\,  \frac{1}{(1-\sin(\alpha))|\tau|} \, \leq\,  \frac{1}{1-\sin(\alpha)}
\cdot\frac{1}{|\lambda|}\, .
\end{equation}

\begin{figure}[!ht]
 \centering
\includegraphics[width=0.5\textwidth]{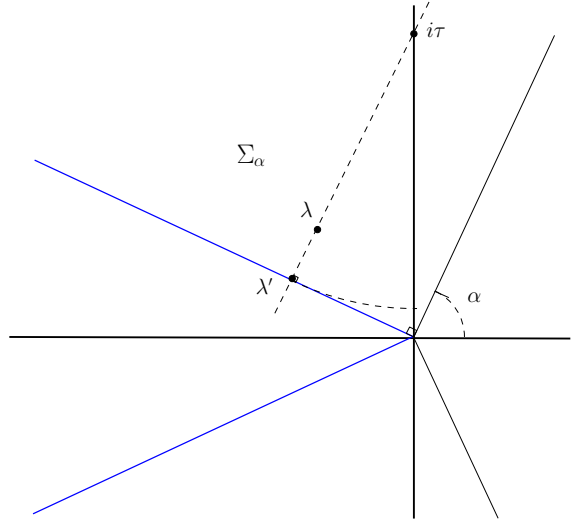}
\caption{The set $\Sigma_\alpha$.}
\label{fig:angle alpha}
\end{figure}


The fact that $T_{A, \gep}(t)$ can be extended to an analytic semigroup $T_{A, \gep}(z)$ on the domain $\Delta_\alpha$ is a simple
application of \eqref{eq:estimRAepsdef} and Proposition \ref{prop:pazysemgps}, with $M:=\frac{1}{1-\sin(\alpha)}$.

Let us then define $\tilde{T_A}(z):= e^{\gep z} T_{A,\gep}(z)$, for $z\in\Delta_\alpha$
so that $\tilde{T_A}$ is an
analytic extension of $T_A$ (an argument of analyticity shows that $\tilde{T_A}$ does not depend on $\gep$).
\end{proof}

\begin{rem}\rm
Note that estimate \eqref{eq:estim R lambda alpha} is also valid in the limit as $\gep\to 0$: for all $\alpha\in(0, \frac\pi2)$,
$\lambda\in\Sigma_\alpha$,
\begin{equation}
\label{eq:estim R lambda alpha without gep}
 \Nsqmu{R(\lambda, A)}\, \leq\, \frac{1}{1-\sin(\alpha)}\cdot \frac{1}{|\lambda|}\, .
\end{equation}
\end{rem}

\subsubsection{Spectral properties of $L^\go_{{q}}=A+B$}

In this part, we show that if the perturbation $B$ is small enough with respect to $A$, one has the same spectral properties for
$L^\go_{{q}}= A+B$ as for $A$. In this extent, we recall that $\mu$ is of compact support in $[-1, 1]$, and the disorder is
rescaled by  $\gd>0$.

\begin{proposition}
\label{prop:BAboundedmu}
 \item The operator $B$ is $A$-bounded, in the sense that there exist explicit constants $a_{K, \gd}$ and $b_{K,\gd}$, depending
on $K$ and $\gd$ such that for all $u$ in the domain of (the closure of) $A$
\begin{equation}
 \label{eq:BAboundedmu}
\Nsqmu{Bu}\, \leq\, a_{K, \gd}\Nsqmu{u} +  b_{K, \gd}\Nsqmu{Au}\, .
\end{equation}
Moreover, for fixed $K>1$, $a_{K,\gd}=O(\gd)$ and $b_{K,\gd}=O(\gd)$, as $\gd\rightarrow0$.
\end{proposition}

The latter proposition is based on the fact that the difference $\gep(q)(\gtta, \go, \gd) =  q(\gtta, \gd\go)- q_0(\gtta)$
in \eqref{eq:defdeltaqmu} is small if the scale parameter $\gd$ tend to $0$:
\begin{lemma}
 \label{lem:estimdeltaq}
For $\gd>0$, let us define

\begin{equation}
\label{eq:def Ninf eps(q)}
 \Ninf{\gep(q)}\, :=\, \suptwo{\gtta\in\bbS , |\go|\leq 1}{0<u<\gd} |\gep(q)(\gtta, \go, u)|\, .
\end{equation}
Then for all $K>1$, $\Ninf{\gep(q)} = O(\gd)$, as $\gd\rightarrow 0$. More precisely, for $K>1$, $\gd>0$, the following
inequality holds:
\begin{equation}
\label{eq:estimdeltaq} 
\Ninf{\gep(q)}\, \leq\, \gep_{K,\gd}\, ,
\end{equation}
where the constant $\gep_{K, \gd}$ can be chosen explicitly in terms of $K$ and $\gd$:
\begin{equation}
 \label{eq:estimdeltaK}
\gep_{K, \gd}\, :=\, \frac{\gd}{\pi} e^{8\pi\gd}\left( 2+ 3e^{4\pi\gd} \right) e^{14K\bar{r}_\gd}\left( 1+ 2\pi
e^{2K\bar{r}_\gd}\right)\, ,
\end{equation}
where we recall that $\bar{r}_\gd= \max\left( r_0, r_\gd\right)$.
\end{lemma}

\begin{proof}[Proof of Lemma \ref{lem:estimdeltaq}]
 Recall that the disordered stationary solution $q$ \eqref{eq:qhatom} is
given by 
\begin{equation}
 q(\gtta, \gd\go) \,:=\, \frac{S(\gtta, \gd\go, 2Kr_\gd)}{Z(\gd\go, 2Kr_\gd)}, 
\end{equation}
where $S(\gtta, \go, x)$ is defined in
\eqref{eq:Sq} and that the non-disordered one
\eqref{eq:defstationarysolution
nodisorder} is given by $q_0(\gtta)=\frac{S(\gtta, 0, 2Kr_0)}{Z(0, 2Kr_0)}= \frac{e^{2Kr_0\cos(\gtta)}}{\intS
e^{2Kr_0\cos(\gtta)}\dd\gtta}$. 
Since $q(\gtta, \gd\go)=q(-\gtta, -\gd\go)$, it suffices to consider the case $\gd\go>0$.
A simple computation shows that 
\begin{equation}
 \label{eq:lower bound Zomega}
Z(\gd\go, 2Kr_\gd)\, \geq\, 4\pi^2 e^{-4K r_\gd}e^{-4\pi\gd}\, ,
\end{equation}
and that
\begin{equation}
 \label{eq:upper bound S0}
|S(\gtta, 0)|\, \leq\, 2\pi e^{4Kr_0}\, .
\end{equation}

Using $|q(\gtta, \gd\go)- q_0(\gtta)|\leq \frac{1}{Z(\gd\go)Z(0)}\left( Z(0)|S(\gtta, \gd\go) - S(\gtta, 0)| +
|S(\gtta, 0)||Z(0)-Z(\gd\go)|\right)$, one has to deal with, successively:
\begin{itemize}
 \item for fixed $\gtta\in\bbS$, $|S(\gtta, \gd\go) - S(\gtta, 0)|\leq \gd\cdot \sup_{|\go|\leq 1}|\frac{\dd}{\dd\go} S(\gtta,
\gd\go)|$. A long calculation shows that the latter expression $|\frac{\dd}{\dd\go} S(\gtta, \gd\go)|$ can be bounded above by
$8\pi^2 e^{4Kr_\gd}e^{4\pi\gd}\left( 2+ 3e^{4\pi\gd}\right)$, that is, 
\begin{equation}
 \label{eq:upper bound diff S}
|S(\gtta, \gd\go) - S(\gtta, 0)|\, \leq\, \gd 8\pi^2 e^{4Kr_\gd}e^{4\pi\gd}\left( 2+ 3e^{4\pi\gd}\right)\, .
\end{equation}
\item Using $|Z(\gd\go) - Z(0)|= \left|\intS (S(\gtta, \gd\go) - S(\gtta, 0)) \dd\gtta\right|$ and \eqref{eq:upper bound diff S},
one has directly:
\begin{equation}
 \label{eq:upper bound diff Z}
|Z(\gd\go) - Z(0)|\, \leq\, \gd 16\pi^3 e^{4Kr_\gd}e^{4\pi\gd}\left( 2+ 3e^{4\pi\gd}\right)\, .
\end{equation}
\end{itemize}
Putting together \eqref{eq:lower bound Zomega}, \eqref{eq:upper bound S0}, \eqref{eq:upper bound diff S} and \eqref{eq:upper bound
diff Z}, one obtains the result.
\end{proof}

We are now in position to prove the $A$-boundedness of $B$:
\begin{proof}[Proof of Proposition \ref{prop:BAboundedmu}]
$B$ is $A$-bounded: let us fix a $u$ in the domain of the closure of $A$. Then we have $\Nsqmu{Bu}= \NLqmu{\cB u}$, where $\cB u$ is
the
appropriate primitive of $Bu$, namely:
\begin{align}
\cB u(\gtta, \go)&\, :=\, {} -\left( u(\gtta, \go) \langle J \ast \gep(q) \rangle_\mu + \gep(q)(\gtta, \go, \gd) \langle J \ast
u\rangle_\mu(\gtta) + \gd\go u(\gtta, \go)\right)\nonumber\\ 
&+{} \left( \intS \frac{1}{q_0} \right)^{-1}\left( \intS \frac{u(\gtta, \go) \langle J \ast \gep(q) \rangle_\mu + \gep(q)(\gtta,
\go, \gd) \langle J \ast
u\rangle_\mu(\gtta) + \gd\go u(\gtta, \go)}{q_0(\gtta)} \dd\gtta\right)\, .\label{eq:primitiveBu}
\end{align}
 
One can easily shows that there exists a constant $c^{(1)}_{K, \gd}$, depending only on $K>1$ and $\gd>0$ such
that:
\begin{equation}
 \label{eq:estim1Bu}
\Nsqmu{Bu} \, \leq\,  c^{(1)}_{K, \gd} \NLqmu{u}\, .
\end{equation}
Indeed, an easy calculation shows that $|\langle J\ast \gep(q)\rangle_\mu|\leq 4K\Ninf{\gep(q)}$ and that 
\begin{equation}
\begin{split}
|\langle J\ast
u\rangle_\mu(\cdot)|&\, \leq\, K\left( \intS \sin(\cdot-\varphi)^2 q_0(\varphi)\dd\varphi \right)^{\frac12}\NLqmu{u}\\
&\,\leq K\left( \intS q_0(\varphi)\dd\varphi \right)^{\frac12}\NLqmu{u} = K \NLqmu{u}\, . 
\end{split}
\end{equation}
So we have for all $\gtta, \go$ (recall that $Z_0$ is
the normalization constant in
\eqref{eq:defstationarysolution nodisorder}):
\begin{equation}
\begin{split}
 |\cB u(\gtta, \go)|\, \leq\, &  \left(4 K\Ninf{\gep(q)}+\gd|\go|\right)|u|+ 2K \Ninf{\gep(q)}\NLqmu{u}\\ &+
Z_{0}^{-1}\left(4 K
\Ninf{\gep(q)} + \gd|\go|\right) \left( \intS \frac{|u|^2}{q_0} \right)^{\frac 12}\, .
\end{split}
\end{equation}
Hence, inequality \eqref{eq:estim1Bu} is true for the following choice of $c^{(1)}_{K, \gd}$ (recall that $\gep_{K, \gd}$ is defined in
\eqref{eq:estimdeltaK}):
\begin{equation}
 \label{eq:defc1}
c^{(1)}_{K, \gd}\, :=\,  \left( 6\left(4 K\gep_{K, \gd}+\gd\right)^2 +12 K^2 Z_{0}^2 \gep_{K, \gd}^2 \right)^{\frac
12}\, .
\end{equation}

\begin{rem}\rm
 Note that, thanks to Lemma \ref{lem:estimdeltaq}, one has that $c^{(1)}_{K, \gd}=O(\gd)$ as $\gd\rightarrow0$.
\end{rem}

In order to complete the proof of the inequality \eqref{eq:BAboundedmu}, it suffices to prove that there exist constants
$c^{(2)}_K$
and $c^{(3)}_K$, only depending on $K$ such that, for all $u$:

\begin{equation}
 \label{eq:estim2Bu}
\NLqmu{u} \, \leq\,  c^{(2)}_K \Nsqmu{Au} + c^{(3)}_K \Nsqmu{u}\, .
\end{equation}

The rest of this first of the proof is devoted to find explicit expressions of $c^{(2)}_K$ and $c^{(3)}_K$, and is based on an
interpolation argument.  

For all integer $n>1$, one can compute the linear operator $f\mapsto f'$ in terms of a sum of two integral operators,
namely:
\begin{equation}
\label{eq:identuprimemu}
f' \, =\,  \In(f'') + \Jn(f)\, , 
\end{equation}
where $\In:f\mapsto \int_{0}^{2\pi} i_n(\gtta, \gp)f(\gp)\dd\gp$ (resp.
$\Jn:f\mapsto \int_{0}^{2\pi} j_n(\gtta, \gp)f(\gp)\dd\gp$) is the integral operator whose kernel $i_n(\gtta,
\gp)$ (resp. $j_n(\gtta, \gp)$) is defined by:
\begin{equation}\left\{\begin{array}{lll}
   i_n(\gtta, \gp)\, :=\,  \frac{\gp^{n+1}}{2\pi\gtta^{n}}\, ,& j_n(\gtta, \gp)\, :=\, -\frac{n(n+1)\gp^{n-1}}{2\pi\gtta^n}\, ,&
0\, \leq\,  \gp\, <\, \gtta\, \leq\,  2\pi\, ,\\[10pt]
i_n(\gtta, \gp)\, :=\,  \frac{-(2\pi-\gp)^{n+1}}{2\pi(2\pi-\gtta)^n}\ , & j_n(\gtta, \gp)\, :=\,  \frac{n(n+1)
(2\pi-\gp)^{n-1}}{2\pi (2\pi-\gtta)^n}\ , & 0\, \leq\,  \gtta\, <\, \gp\, \leq\,  2\pi\, .
  \end{array}\right.\end{equation}
Equality \eqref{eq:identuprimemu} can be easily verified by integrations by parts. Since,
\begin{equation}\left\{\begin{array}{ll}
  \int_{0}^{2\pi} \left|i_n(\gtta, \gp)\right|\dd\gp \, \leq\,  \frac{2\pi}{n+2},& \int_{0}^{2\pi} \left|i_n(\gtta,
\gp)\right|\dd\gtta\, \leq\,  \frac{2\pi}{n-1}\, ,\\[10pt]
  \int_{0}^{2\pi} \left|j_n(\gtta, \gp)\right|\dd\gp \, \leq\,  \frac{n+1}{\pi}\, ,& \int_{0}^{2\pi} \left|j_n(\gtta,
\gp)\right|\dd\gtta\, \leq\,  \frac{n(n+1)}{\pi(n-1)}\, ,
  \end{array}\right.\end{equation}
we see (cf. \cite[p.143-144]{Kato1995}) that $\In$ and $\Jn$ are bounded operators on $L^2({\bbS })$, namely:
\begin{equation}\Vert{\In}\Vert\, \leq\,  \frac{2\pi}{n-1},\quad \Vert{\Jn}\Vert\, \leq\,  \frac{n(n+1)}{\pi(n-1)}\, .\end{equation}
So, applying relation \eqref{eq:identuprimemu} for $f=\cU$ we get, for $\mu$-almost every $\go$:
\begin{equation}
\left( \intS |u(\gtta, \go)|^2 \dd\gtta \right)^{\frac 12} \leq\frac{2\pi}{n-1}\left( \intS |u'(\gtta, \go)|^2 \dd\gtta
\right)^{\frac 12} +
\frac{n(n+1)}{\pi(n-1)}\left(\intS|\cU(\gtta, \go)|^2 \dd\gtta \right)^{\frac 12}. 
\end{equation}

This gives
\begin{equation}
\label{eq:estuUmu} 
\NLmu{u}\, \leq\,\frac{2\pi}{n-1}\NLmu{u'} +\frac{n(n+1)}{\pi(n-1)}\NLmu{\cU}\, .
\end{equation}
Since $\NLqmu{\cU}= \Nsqmu{u}$, it only remains to control $\NLqmu{u'}$ with $\Nsqmu{Au}$:
like for the beginning of this proof for the operator $B$, we have $\Nsqmu{Au}=  \NLqmu{\cA u}$, where $\cA u$ is the
appropriate primitive of $Au$:
\begin{align}
\cA u(\gtta, \go)&\, :=\, {} {\frac 12} u'(\gtta, \go) -\left( u(\gtta, \go) (J \ast q_0) + q_0(\gtta) \langle J
\ast u\rangle_\mu(\gtta)\right)\nonumber\\ 
&+{} \left( \intS \frac{1}{q_0} \right)^{-1}\left( \intS \left\{\frac{u(\gtta, \go) (J \ast
q_0)}{q_0(\gtta)} + {\frac 12} u(\gtta, \go) \partial_\gtta\left( \frac{1}{q_0(\gtta)}
\right)\right\}\dd\gtta\right)\, .\label{eq:primitiveAu}
\end{align}
Using inequalities $|\langle J\ast u\rangle|_{\mu}(\cdot)\leq K\sqrt{\pi}\NLmu{u}$, and $\intS \frac{|u(\cdot, \go)|}{q_0}\leq
Z_0^{\frac12}e^{Kr_0}\left( \intS |u(\cdot, \go)^2|\right)^{\frac12}$, an easy calculation shows that:
\begin{equation}
\label{eq:estimuprimeAu1}
 |u'(\cdot, \go)|\, \leq\, 2|\cA u(\cdot, \go)| + 2Kr_0|u(\cdot, \go)|+2\sqrt{\pi}K q_0(\cdot)\NLmu{u} +
\frac{4Kr_0}{Z_0^\frac12}e^{Kr_0}\left( \intS |u(\cdot, \go)^2|\right)^{\frac12}\, ,
\end{equation}
and thus,
\begin{equation}
\label{eq:estimuprimeAu}
 \NLmu{u'}\, \leq\, 4 \NLmu{\cA u} + 4K \left( r_0^2+ \pi Z_0^{-1}e^{2Kr_0}(1+8r_0^2) \right)^{\frac 12}\NLmu{u}\, ,
\end{equation}
and by
putting  \eqref{eq:estuUmu} and \eqref{eq:estimuprimeAu} together we obtain
\begin{equation}
\begin{split}
\NLmu{u} \, \leq\,&  \frac{8\pi}{n-1}\NLmu{\cA u} + \frac{2\pi}{n-1}4K \left( r_0^2+ \pi Z_0^{-1}e^{2Kr_0}(1+8r_0^2) \right)^{\frac 12}\NLmu{u}\\
&+\frac{n(n+1)}{\pi(n-1)}\Nsqmu{u}\, .
\end{split}
\end{equation}
Let us choose the integer $n=\left\lfloor 16\pi K\left( r_0^2+ \pi Z_0^{-1}e^{2Kr_0}(1+8r_0^2) \right)^{\frac 12} +1\right\rfloor$ so that
\begin{equation}\frac{2\pi}{n-1}4K \left( r_0^2+ \pi Z_0^{-1}e^{2Kr_0}(1+8r_0^2) \right)^{\frac 12}\, \leq\,  {\frac 12}\, .\end{equation} In
this
case, we
obtain:
\begin{align}
 \NLqmu{u} &\, \leq\,  \frac{e^{2Kr_0}}{4K\left( r_0^2+ \pi Z_0^{-1}e^{2Kr_0}(1+8r_0^2) \right)^{\frac 12}}\Nsqmu{Au}\nonumber\\&+
\frac{e^{2Kr_0}\left(16 K\left(r_0^2+ \pi Z_0^{-1}e^{2Kr_0}(1+8r_0^2)\right)^{\frac 12}
+3\right)^2}{16\pi^2 K\left(r_0^2+ \pi Z_0^{-1}e^{2Kr_0}(1+8 r_0^2)\right)^{\frac 12}}\Nsqmu{u}\, ,
\end{align}
which is precisely the inequality \eqref{eq:estim2Bu} we wanted to prove. Inequalities \eqref{eq:estim1Bu} and
\eqref{eq:estim2Bu} give the result, for $a_{K, \gd}:= c^{(1)}_{K, \gd} \cdot c_K^{(3)}$ and $b_{K, \gd}:= c^{(1)}_{K, \gd} \cdot
c_K^{(2)}$.\qedhere
\end{proof}

\begin{proposition}
 \label{prop:Lcompactresolvent}
For all $K>1$, there exists $\gd_{3}(K)>0$ such that for all $0<\gd\leq\gd_{3}(K)$, the operator $L^\go_{{q}}$ is closable. In
that case, its closure has the same domain as the closure of $A$.
\end{proposition}

\begin{proof}
Let us choose $\gd_3(K)>0$ so that
\begin{equation}
\label{eq:condomega37}
b_{K, \gd_3(K)}<1
\end{equation}
where $b_{K, \gd}$ is the constant introduced  in \eqref{eq:BAboundedmu}, then, for
all $0<\gd\leq\gd_3(K)$, the operator $B$ is $A$-bounded with $A$-bound strictly lower than $1$. The result is then a
consequence of Th. IV-1.1, p.190 in \cite{Kato1995}.
\end{proof}

\subsubsection{The spectrum of $L^\go_{{q}}$}

We divide our study into two parts: the determination of the position of the spectrum within a sector and its position near $0$.

\medskip
\subsubsection{Position of the spectrum away from $0$} We prove mainly that the perturbed operator $L^\go_{{q}}$ still generates an analytic
semigroup of operators on an appropriate sector. An immediate corollary is the fact that the spectrum lies in a cone whose vertex is zero.

We know (Proposition \ref{prop:Apositivemu}) that for all $0<\alpha<\frac{\pi}{2}$, $A$ generates an analytic 
semigroup of operators on $\Delta_\alpha:=\ens{\lambda\in\bC}{|\arg(\lambda)|<\alpha}$.

\begin{proposition}
 \label{prop:semigroupAB}
For all $K>1$, $0<\alpha<\frac{\pi}{2}$ and $\gep>0$, there exists $\gd_4>0$ (depending on $\alpha$, $K$ and $\gep$) such that for
all $0<\gd<\gd_{4}$, the spectrum of $L^\go_{{q}}= A+B$ lies within $\Theta_{\gep,
\alpha}:= \ens{\lambda\in\bC}{\frac{\pi}{2} + \alpha\leq
\arg(\lambda)\leq \frac{3\pi}{2}-\alpha}\cup \ens{\lambda\in\bC}{|\lambda|\leq \gep}$. Moreover, there exists $\alpha'\in(0, \frac\pi2)$ such
that the operator $L^\go_{{q}}$ still generates an analytic semigroup on $\Delta_{\alpha'}$.
\end{proposition}
\begin{proof}[Proof of Proposition \ref{prop:semigroupAB}]
Let $0<\alpha<\frac{\pi}{2}$ be fixed.
Following \eqref{eq:BAboundedmu} and using \eqref{eq:estim R lambda alpha without gep}, one can easily deduce an estimate on the bounded operator
$BR(\lambda, A)$, for $\lambda\in\Sigma_\alpha$:
\begin{equation}
\begin{split}
 \Nsqmu{BR(\lambda, A)u} &\,\leq\, a_{K,\gd}\Nsqmu{R(\lambda, A)u} + b_{K,
\gd}\Nsqmu{A R(\lambda,A)u} \\
&\, \leq\, a_{K, \gd}\frac{1}{(1-\sin(\alpha))|\lambda|}\Nsqmu{u}\\ &
\ \ \ \ \ \ 
+ b_{K, \gd}\left(1+
\frac{1}{1-\sin(\alpha)}\right) \Nsqmu{u}\, .
\end{split}
\end{equation}
Let us fix $\gep>0$ and choose $\gd$ so that:
\begin{equation}
\label{eq:condomega00}
\max\left(4b_{K, \gd}\left(\frac{1}{1-\sin(\alpha)}+1\right) , \frac{4 a_{K, \gd}}{(1-\sin(\alpha))\gep} \right)  \,
\leq\,  1\, .
\end{equation}

Then, for $\lambda\in\Sigma_\alpha$ such that $|\lambda|>\gep\geq \frac{4a_{K, \gd}}{1-\sin(\alpha)}$, we have
\begin{equation}\Nsqmu{BR(\lambda, A)u} \, \leq\,  \frac{1}{2}\Nsqmu{u}\, .\end{equation}
In particular, $1 - BR(\lambda, A)$ is invertible with $\Nsqmu{\left( 1 - BR(\lambda, A) \right)^{-1}}\leq 2$.
A direct calculation shows that \begin{equation}\left( \lambda - (A+B) \right)^{-1} \, =\,  R(\lambda, A) \left( 1 -
BR(\lambda, A)\right)^{-1}\, .\end{equation} One deduces the following estimates on the resolvent: for $\lambda\in\Sigma_\alpha$,
$|\lambda|>\gep$,
\begin{equation}\label{eq:estim resolvent Lq gep prime}\Nsqmu{R(\lambda, L^\go_{{q}})}\, \leq\,  \frac{2}{(1-\sin(\alpha))|\lambda|}\,
.\end{equation}
Estimate \eqref{eq:estim resolvent Lq gep prime} has two consequences: firstly, one deduces immediately that the spectrum $\sigma(L^\go_{q})$ of
$L^\go_{q}$ is contained in $\Theta_{\gep,\alpha}$:
\begin{equation}
\label{eq:subset rho Lq}
 \sigma(L_q^{\go})\subseteq\ens{\lambda\in\bC}{\frac{\pi}{2} + \alpha\leq \arg(\lambda)\leq \frac{3\pi}{2}-\alpha}\cup
\ens{\lambda\in\bC}{|\lambda|\leq \gep}.
\end{equation}
Secondly, \eqref{eq:estim resolvent Lq gep prime} entails that $L^\go_{q}$ generates an analytic semigroup of operators on an appropriate
sector. Indeed, if one denotes by $L^{\go}_{q, \gep}:= L^{\go}_{q} -\gep$, one deduces from \eqref{eq:subset rho Lq} that
$0\in\rho(L^{\go}_{q, 2\gep})$ and that for all $\lambda\in\bC$ with $\Re(\lambda)>0$ (in particular, $|\lambda|<|\lambda+2\gep|$)
\begin{align}
 \Nsqmu{R(\lambda, L^\go_{{q,2\gep}})}\, &=\, \Nsqmu{R(\lambda+2\gep, L^\go_{{q}})}\, \leq\, \frac{2}{(1-\sin(\alpha))|\lambda+2\gep|}\,
,\nonumber\\
&\, \leq\, \frac{2}{(1-\sin(\alpha))|\lambda|}\, .
\end{align}
Hence, using the same arguments of Taylor expansion as in the proof of Proposition \ref{prop:Apositivemu} and applying Proposition
\ref{prop:pazysemgps}, one easily sees that $L^\go_{{q,2\gep}}$ generates an analytic semigroup in a (a priori) smaller sector
$\Delta_{\alpha'}$, where $\alpha'\in(0, \frac\pi2)$ can be chosen as $\alpha':= \frac12 \arctan\left( \frac{1-\sin(\alpha)}{2} \right)$. But if
$L^\go_{{q,2\gep}}$ generates an analytic semigroup, so does $L^\go_{q}$.
\end{proof}

\subsubsection{Position of the spectrum near $0$}
\label{subsubsec:loczero}
Let us apply Proposition \ref{prop:semigroupAB} for fixed $K>1$, $\alpha\in(0, \frac\pi2)$, $\rho\in(0,1)$ and $\gep:=\rho\gap$, where we recall
that $\gap$ is the spectral gap between the eigenvalue $0$ for the non perturbed operator $A$ and the rest of the spectrum
$\sigma(A)\smallsetminus\{0\}$.
Let $\Theta_{\gep, \alpha}^{+}:=  \ens{\lambda\in\Theta_{\gep, \alpha}}{\Re(\lambda)\geq 0}$ be the subset of
$\Theta_{\gep, \alpha}$ which lies in the positive part of the complex plane (see Fig. \ref{fig:Thetaeps}). In order to show
the linear stability, one has
to make sure that one can choose a perturbation $B$ small enough so that no eigenvalue of $A+B$ remains in the small set
$\Theta_{\gep, \alpha}^{+}$.

\begin{figure}[!ht]
 \centering
\includegraphics[width=0.5\textwidth]{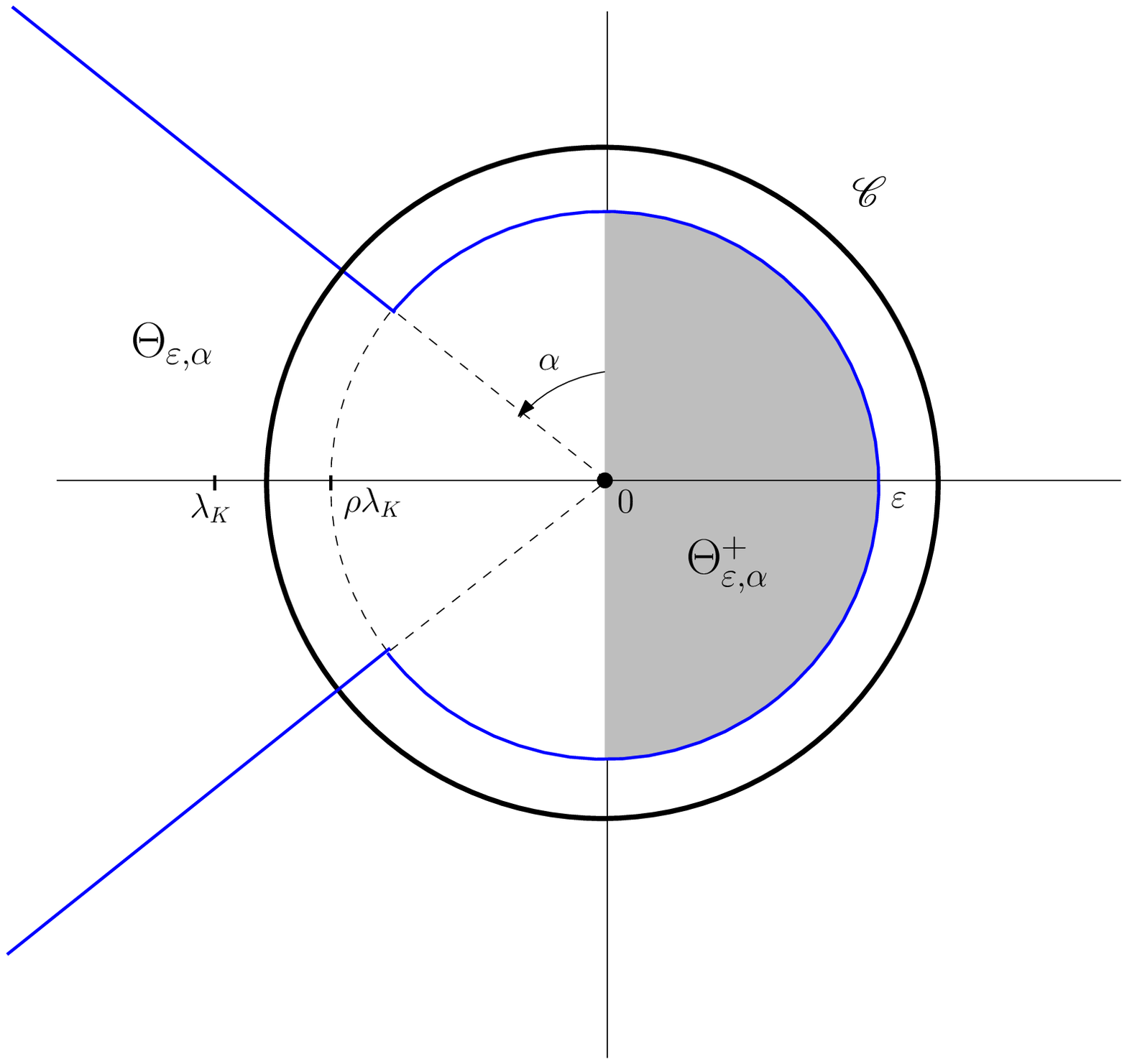}
\caption{The set $\Theta_{\gep,
\alpha}$.}
\label{fig:Thetaeps}
\end{figure}

Since $\gap>0$, one can separate $0$ from the rest of the spectrum of $A$ by a circle $\mathscr{C}$ centered in $0$ with radius
$(\frac{\rho+1}{2})\gap$. The appropriate choice of $\gep$ ensures that the
interior of the disk delimited by $\mathscr{C}$ contains $\Theta_{\gep, \alpha}^{+}$ (see Figure~\ref{fig:Thetaeps}).

The main argument is the following: by construction of $\mathscr{C}$, $0$ is the only eigenvalue (with multiplicity $1$) of the
non-perturbed operator $A$ lying in the interior of $\mathscr{C}$. A principle of local continuity of eigenvalues shows that, while
adding a sufficiently small perturbation $B$ to $A$, the interior of $\mathscr{C}$ still contains exactly one eigenvalue (which is \emph{a
priori}
close but not equal to $0$) with the same multiplicity.

But we already know that for the perturbed operator $L^\go_{{q}}=A+B$, $0$ is always an eigenvalue (since $L^\go_{{q}}q'=0$). One
can therefore
conclude that, by uniqueness, $0$ is the only element of the spectrum of $L^\go_{{q}}$ within $\mathscr{C}$, and is an eigenvalue with
multiplicity $1$. In particular, there is no element of the spectrum in the positive part of the complex plane.

In order to quantify the appropriate size of the perturbation $B$, one has to have explicit estimates on the resolvent $R(\lambda,
A)$ on the
circle $\mathscr{C}$. 

\begin{lemma}
 \label{lem:estimRGamma}
There exists some explicit constant $c_{\mathscr{C}}=c_{\mathscr{C}}(K, \rho)$ such that for all $\lambda\in\mathscr{C}$,
\begin{align}
 \Nsqmu{R(\lambda, A)} &\, \leq\,  c_{\mathscr{C}}\, ,\label{eq:estimRcercle}\\
\Nsqmu{AR(\lambda, A)} &\, \leq\,  1+ \left(\frac{1+\rho}{2}\right)\gap\cdot c_{\mathscr{C}}\, .\label{eq:estimARcercle}
\end{align}
One can choose $c_{\mathscr{C}}$ as $\frac{1}{\gap} \max\left( \frac{2}{\rho+1},\, \frac{2}{1-\rho}\right):= \frac{\ell(\rho)}{\gap}$.
\end{lemma}
\begin{proof}[Proof of Lemma \ref{lem:estimRGamma}]
Applying the spectral theorem (see \cite[Th. 3, p.1192]{cf:Dunford}) to the essentially self-adjoint operator $A$, there exists a
spectral measure $E$ vanishing on the complementary of the spectrum of $A$ such that $A=\intR \lambda \dd E(\lambda)$. In that
extent, one has for any $\zeta\in\mathscr{C}$
\begin{equation}
 R(\zeta, A)\, =\, \intR \frac{\dd E(\lambda)}{\lambda -\zeta}\, .
\end{equation}
In particular, for $\zeta\in\mathscr{C}$
\begin{equation}
 \Nsqmu{R(\zeta, A)} \leq \sup_{\lambda\in \sigma(A)} \frac{1}{|\lambda-\zeta|}\leq \frac{\ell(\rho)}{\gap}\, .
\end{equation}

The estimation \eqref{eq:estimARcercle} is straightforward.
\end{proof}

We are now in position to apply our argument of local continuity of eigenvalues: Following \cite[Th III-6.17,
p.178]{Kato1995}, there exists a decomposition of the
operator $A$ according to $H^{-1}_{{q_0},\mu}= H_0 \oplus H'$ (in
the
sense that $AH_0\subset H_0$, $A H'\subset H'$ and $P \cD(A) \subset \cD(A)$,
where $P$ is the projection on $H_0$ along $H'$) in such a way that $A$ restricted to $H_0$ has spectrum
$\{0\}$ and $A$ restricted to $H'$ has spectrum $\sigma(A)\smallsetminus\{0\}$. 

Let us note that the dimension of $H_0$ is $1$, since the characteristic space of $A$ in the eigenvalue $0$ is reduced to its
kernel which is of dimension $1$.

Then, applying \cite[Th. IV-3.18, p.214]{Kato1995}, and using Proposition \ref{prop:BAboundedmu},
we find that if one chooses $\gd>0$, such that
\begin{equation}
\label{eq:condperturbGamma}
 \sup_{\lambda\in\mathscr{C}} \left(a_{K, \gd} \Nsqmu{R(\lambda, A)} + b_{K, \gd}\Nsqmu{AR(\lambda,
A)}\right)<1,
\end{equation}
then the perturbed operator $L^\go_{{q}}$ is likewise decomposed according to
$H^{-1}_{{q_0},\mu}= \tilde{H}_0\oplus \tilde{H}'$, in such a way that $\dim(H_0)= \dim(\tilde{H}_0)= 1$, and that the spectrum of
$L^\go_{{q}}$
is again separated in two parts by $\mathscr{C}$ . But we already know that the characteristic space of the
perturbed operator $L^\go_{{q}}$ according to the eigenvalue $0$ is, at least, of dimension $1$ (since $L^\go_{{q}}q'=0$).

We can conclude, that for such an $\gd>0$, $0$ is the only eigenvalue in $\mathscr{C}$ and that $\dim(\tilde{H}_0)=1$.

Applying Lemma \ref{lem:estimRGamma}, we see that condition \eqref{eq:condperturbGamma} is satisfied if we choose $\gd>0$
so that:
\begin{equation}
 \label{eq:condomega99}
a_{K, \gd} c_{\mathscr{C}} + b_{K, \gd}\left( 1 + \left( \frac{1+\rho}{2} \right)\gap c_{\mathscr{C}}\right)<1.
\end{equation}

In particular, in that case, the spectrum of $L^\go_{{q}}$ is contained in 
\begin{equation}
 \ens{\lambda\in\bbC}{\frac{\pi}{2} + \alpha\leq
\arg(\lambda)\leq \frac{3\pi}{2}-\alpha}\subseteq
\ens{z\in\bC}{\Re(z)\leq 0}\, .
\end{equation}
Finally, the following proposition sums-up the sufficient conditions on $\gd$ for the conclusions of Theorem \ref{th:spectral
prop L disorder} to be satisfied:

\begin{proposition}
 \label{prop:resume conditions delta_2}

Recall the definitions of $a_{K, \gd}$ and $b_{K, \gd}$ in Proposition \ref{prop:BAboundedmu}.
If $\gd>0$ satisfies the following conditions
\begin{equation}
\label{eq:conds4}
\begin{split}
b_{K, \gd} &\, \leq\,  1\, ,\\
4b_{K, \gd}\left( \frac{1}{1-\sin(\alpha)} +1\right) &\, \leq\,  1\, ,\\
\frac{4 a_{K, \gd} }{\rho\gap\left( 1-\sin(\alpha) \right)}&\, \leq\, 1\, ,\\
a_{K, \gd}\frac{\ell(\rho)}{\gap} + b_{K, \gd} \left( 1+ \left(\frac{1+\rho}{2}\right)\ell(\rho)\right)&\, <\, 1\, .
\end{split}
\end{equation}
the conclusions of Theorem \ref{th:spectral prop L disorder} are true.
\end{proposition}
\medskip 

\begin{proof}
 One has simply to sum-up conditions \eqref{eq:condomega37}, \eqref{eq:condomega00} with $\gep=\rho\gap$ and
\eqref{eq:condomega99}. \eqref{eq:estim gd K} can be obtained by (long) estimations on the coefficients $a_{K, \gd}$ and $b_{K, \gd}$.
\end{proof}
\medskip

\begin{rem}\rm
The conditions in Proposition~\ref{prop:resume conditions delta_2}
can be simplified. For example one can exhibit an explicit
  constant $c$ such that if $\gd$
satisfies
\begin{equation}
\label{eq:estim gd K}
\begin{split}
 \gd e^{12\pi\gd}\,\leq\, c e^{-20K\bar{r_\gd}}\max&\left(1, \left( \frac{1-\sin(\alpha)}{2-\sin(\alpha)} \right),
\frac{\rho\gap(1-\sin(\alpha))e^{-4K\bar{r_\gd}}}{K^2},\right.\\
&\left.\ \ \ \ \frac{\gap}{K^2e^{4K\bar{r_\gd}}\ell(\rho) + \gap \left( 1+ \left( \frac{1+\rho}{2} \right)\ell(\rho) \right)}\right)
\end{split}
\end{equation}
the conditions in \eqref{eq:conds4} are fulfilled. Explicit  estimates on the spectral gap $\gl_K$
can be found 
in \cite[Sec.~2.5]{cf:BGP}. 
\end{rem}

\section*{Acknowledgments}
We  are grateful to K. Pakdaman and G. Wainrib for helpful discussions.
G. G. acknowledges the support of the ANR (projects Mandy and SHEPI) and 
the support of the Petronio Fellowship Fund at the Institute for Advanced Study
(Princeton, NJ) where this work has been completed.

\begin{appendix}
 \section{Regularity in the non-linear Fokker-Planck equation}
\label{sec:appendix regularity pt with disorder}
The purpose of this section is to establish regularity properties of the solution of the non-linear
equation \eqref{eq:AR} (where we fix $\gd=1$ for simplicity). Note that this case also captures the situation where $U(\cdot, \go)\equiv
\go$ (evolution \eqref{FKP kuramoto disorder}), as well as the situation where $U(\cdot, \cdot)\equiv 0$ (evolution \eqref{FKP kuramoto disorder
without drift}). In what follows we make the assumption that $U$ is bounded and that for all $\go\in\Supp(\mu)$, $\gtta\mapsto
U(\gtta, \go)\in C^\infty(\bbS; \, \bbR)$ with bounded derivatives.

The existence and uniqueness in $L^2(\gl\otimes\go)$ of a solution to \eqref{eq:AR} can be tackled using Banach fixed point arguments 
(see \cite[Section 4.7]{cf:SellYou}), but one can obtain more regularity from the theory of fundamental solutions of parabolic equations.

More precisely, it is usual to interpret Equation \eqref{eq:AR} as the strong formulation of the weak equation (where $\nu \in\cC([0, T],
\cM_{1}(\bbS\times\bbR))$ and $F$ is any bounded function on $\bbS\times\bbR$ with twice bounded derivatives w.r.t. $\gtta$):
\begin{align}
\intSR F(\gtta, \go)\nu_t(\dd\gtta, \dd\go)&= \intSR F(\gtta, \go)\nu_0(\dd\gtta, \dd\go) +\frac12\int_0^t \intSR F''(\gtta, \go)\nu_s(\dd\gtta,
\dd\go)\dd s\nonumber\\
&+ \int_0^t \intSR F'(\gtta, \go) \left(\intSR J(\gtta-\cdot)\dd\nu_s+ U(\gtta, \go)\right)\nu_s(\dd\gtta,\dd\go)\dd s,
\label{eq:weak formulation qt with U}
\end{align}
where the second marginal (w.r.t. to the disorder $\go$) of the initial condition $\nu_0(\dd\gtta, \dd\go)$ is $\mu(\dd\go)$ so that one can
write
\begin{equation}
 \label{eq:initial condition nu zero}
\nu_0(\dd\gtta, \dd\go) = \nu_0^\go(\dd\gtta) \mu(\dd\go)\, ,
\end{equation}
where $\nu_0^\go$ is a probability measure on $\bbS$, for $\mu$-a.e. $\go$.

As already mentioned, a proof of the existence of a solution on $[0, T]$ of \eqref{eq:weak formulation qt with U} can be obtained
from the almost-sure convergence of the empirical measure of the microscopic system \cite{cf:eric}. One can also find a proof of uniqueness of
such a solution relying on arguments introduced in \cite{cf:Oelschlager}. 

The regularity result can be stated as follows:
\begin{proposition}
 \label{prop:regularity solution qt}
For all probability measure $\nu_0(\dd\gtta, \dd\go)=\nu_0^\go(\dd\gtta)\mu(\dd\go)$ on $\bbS\times\bbR$, for all $T>0$, there exists a unique
solution $\nu$ to \eqref{eq:weak formulation qt with U} in
$\cC([0, T], \cM_1(\bbS\times\bbR))$ such that for all $F\in\cC(\bbS\times\bbR)$,
\begin{equation}
 \lim_{t\searrow0}\intSR F(\gtta, \go)\nu_t(\dd\gtta, \dd\go)= \intSR F(\gtta, \go) \nu_0^\go(\dd\gtta) \mu(\dd\go). 
\end{equation}

Moreover, for all $t>0$, $\nu_t$ is absolutely continuous with respect to $\lambda_1\otimes\mu$ and for $\mu$-a.e. $\go\in Supp(\mu)$, its
density $(t,\gtta, \go)\mapsto p_t(\gtta, \go)$ is strictly positive on $(0, T]\times \bbS$, is $\cC^\infty$ in $(t,
\gtta)$ and solves the Fokker-Planck equation \eqref{eq:AR}.\end{proposition}

\begin{proof}[Proof of Proposition \ref{prop:regularity solution qt}]
Let us fix $T>0$, $\go\in\Supp(\mu)$ and $t\mapsto \nu_t$ the unique solution in $\cC([0, T], \cM_1(\bbS\times\bbR))$ to \eqref{eq:weak
formulation qt with U}. Let us define $R(t, \gtta, \go):= \intSR J(\gtta-\cdot)\dd\nu_t+ U(\gtta, \go)$ and consider the linear equation
\begin{equation}
 \label{eq:linear equation qt with Rt}
\partial_t p_t(\gtta,\go)\, =\, \frac{1}{2} \gD p_t(\gtta,\go) -\partial_\gtta \Big(p_t(\gtta,\go)R(t, \gtta, \go)\Big)\, ,
\end{equation}
such that for $\mu$-a.e. $\go$, for all $F\in\cC(\bbS)$, 
\begin{equation}
\label{eq:linear equation qt with Rt cond limit}
\intS F(\gtta)p_t(\gtta, \go) \dd\gtta\,
\stackrel{t\searrow 0}{\longrightarrow}
 \intS F(\gtta) \nu_0^\go(\dd\gtta)\, . 
\end{equation}

For fixed $\go\in\Supp(\mu)$, $R(\cdot, \cdot, \go)$ is continuous in time and $\cC^\infty$ in $\gtta$.

Suppose for a moment that we have found a weak solution $p_t(\gtta,\go)$ to \eqref{eq:linear equation qt with Rt}-\eqref{eq:linear equation qt
with Rt cond limit} such that for
$\mu$-a.e. $\go$, $p_t(\cdot, \go)$ is strictly positive on $(0,T]\times\bbS$. In particular for such a solution $p$, the quantity
$\intS p_t(\gtta, \go)\dd\gtta$ is conserved for $t>0$, so that $p_t(\cdot, \go)$ is indeed a probability density for all $t>0$. Then
both
probability measures $\nu_t(\dd\gtta, \dd\go)$ and
$p_t(\gtta,\go)\dd\gtta\mu(\dd\go)$ solve

\begin{align}
\intSR F(\gtta, \go)\nu_t(\dd\gtta, \dd\go)&= \intSR F(\gtta, \go)\nu_0(\dd\gtta, \dd\go) +\frac12\int_0^t \intSR F''(\gtta, \go)\nu_s(\dd\gtta,
\dd\go)\dd s\nonumber\\
&+ \int_0^t \intSR F'(\gtta, \go) R(t, \gtta, \go)\nu_s(\dd\gtta,\dd\go)\dd s.
\label{eq:weak formulation qt with Rt}
\end{align}
By \cite{cf:eric} or \cite[Lemma 10]{cf:Oelschlager}, uniqueness in \eqref{eq:weak formulation qt with U} is precisely a consequence
of uniqueness in \eqref{eq:weak formulation qt with Rt}. Hence, by uniqueness in \eqref{eq:weak formulation qt with Rt},
$\nu_t(\dd\gtta, \dd\go)=p_t(\gtta,\go)\dd\gtta\mu(\dd\go)$, which is the result. So it suffices to exhibit a weak solution
$p_t(\gtta,\go)$ to \eqref{eq:linear equation qt with Rt} such that \eqref{eq:linear equation qt with Rt cond limit} is satisfied. 

\medskip

This fact can be deduced from standard results for uniform parabolic PDEs (see \cite{cf:Aronson} and \cite{cf:Friedman} for precise definitions).
In particular, a usual result, which can be found in \cite[\S 7 p.658]{cf:Aronson}, states that \eqref{eq:linear equation qt with Rt} admits a
fundamental solution $\Gamma(\gtta, t; \gtta', s, \go)$ ($t>s$),
which is bounded above and below (see \cite[Th.7, p.661]{cf:Aronson}):
\begin{equation}
 \label{eq:control Gamma}
\frac{1}{C\sqrt{t-s}} \exp\left( \frac{-C(\gtta - \gtta')^2}{\sqrt{t-s}}\right)\leq \Gamma(\gtta, t; \gtta', s, \go) \leq \frac{C}{\sqrt{t-s}}
\exp\left( \frac{-(\gtta - \gtta')^2}{C\sqrt{t-s}}\right)\, .
\end{equation}
Note that the constant $C>0$ only depends on $T$ and the \emph{structure} of the linear operator in \eqref{eq:linear equation qt with
Rt} (see \cite[Th.7, p.661]{cf:Aronson} and \cite[\S 1, p.615]{cf:Aronson}). In particular, since $(\gtta, \go)\mapsto U(\gtta, \go)$ is
bounded, this constant does not depend on $\go$.

Note that the proof given in \cite{cf:Aronson} is done for $\gtta\in\bbR$ but can be readily adapted to our case ($\gtta\in\bbS$). 

Moreover, thanks to Corollary 12.1, p.690 in \cite{cf:Aronson}, the following expression of $p_t(\gtta, \go)$
\begin{equation}
\label{eq:pt Gamma}
p_t(\gtta, \go)=\intS \Gamma(\gtta, t; \gtta', 0, \go)\nu_0^\go(\dd\gtta')
\end{equation}
defines a weak solution of \eqref{eq:linear equation qt with Rt} on $(0,T]\times\bbS$ (namely a weak solution on $(\tau, T]\times \bbS$, for all
$0<\tau<T$) such that \eqref{eq:linear equation qt with Rt cond limit} is satisfied. The positivity and boundedness of $p_t(\cdot, \go)$ for
$t>0$ is an easy consequence of \eqref{eq:control Gamma}. The smoothness of $p_\cdot(\cdot, \go)$ on $(0, T]\times \bbS$ can be derived by
standard bootstrap methods.
\end{proof}

We focus now on the regularity of the solution $p_t(\gtta, \go)$ of \eqref{eq:AR} with respect to the disorder $\go$. We assume here that the
initial condition $\nu_0$ is such that for all $\go\in\Supp(\mu)$, $\nu_0^\go(\dd\gtta)$ is absolutely continuous with respect to the Lebesgue
measure $\lambda_1$ on $\bbS$: there exists a positive integrable function $\gamma(\cdot, \go)$ of integral $1$ on $\bbS$ such that
$\nu_0^\go(\dd\gtta)=\gamma(\gtta, \go)\dd\gtta$. Then we have

\begin{lemma}[Regularity w.r.t. the disorder]
\label{lem:regularity density disorder}
For every $(t_0, \gtta_0)\in(0, \infty)\times\bbS$, for every $\go_0$ which is an accumulation point in $\Supp(\mu)$ such that the following
holds
\begin{equation}
\label{eq:condition continuity gamma}
\intS |\gamma(\gtta, \go)-\gamma(\gtta, \go_0)|\dd\gtta \to 0, \quad\text{as $\go\to\go_0$}\, ,
\end{equation}
then the solution $p$ of \eqref{eq:AR} defined on $(0, \infty)\times\bbS\times \Supp(\mu)$ is continuous at the point $(t_0, \gtta_0, \go_0)$.
\end{lemma}

\begin{proof}[Proof of Lemma \ref{lem:regularity density disorder}]
 For any $\go$ in the support of $\mu$, let for all $t>0$, $\gtta\in\bbS$
\begin{equation}
 \label{eq:def u diff q1 q2}
u(t, \gtta, \go) := p_t(\gtta, \go) - p_t(\gtta, \go_0),
\end{equation}
where $(p_t(\cdot, \cdot))_{t\geq 0}$ is the unique solution of \eqref{eq:linear equation qt with Rt}. It is easy to see that $u$ is a strong
solution to the following PDE
\begin{equation}
 \label{eq:pde verified by u}
\partial_t u(t, \gtta, \go) - \left[\frac12 \gD u(t, \gtta) - \partial_\gtta\left( u(t, \gtta) R(t, \gtta, \go_0)\right)\right] =
\mathcal{R}(t, \gtta, \go),
\end{equation}
where $\mathcal{R}(t, \gtta, \go):= \partial_\gtta\left[ p_t(\gtta, \go)\left( R(t, \gtta, \go) -R(t, \gtta, \go_0) \right)\right]$ and
with initial condition (since $\nu_0^\go(\dd\gtta)=\gamma(\gtta, \go)\dd\gtta$ for all $\go$)
\begin{equation}
 \label{eq:initial condition u}
u(t, \gtta, \go)|_{t\searrow0} =\gamma(\gtta, \go) - \gamma(\gtta, \go_0).
\end{equation}
Then applying \cite[Th. 12 p.25]{cf:Friedman}, $u(t, \gtta, \go)$ can be expressed as 
\begin{equation}
 \label{eq:u Gamma}
u(t, \gtta, \go)=\intS \Gamma(\gtta, t; \gtta', 0, \go_0)(\gamma(\gtta, \go) - \gamma(\gtta, \go_0))\dd\gtta' - \int_0^t \intS \Gamma(\gtta,
t;\gtta', s, \go_0)\mathcal{R}(s, \gtta', \go)\dd\gtta'\dd s.
\end{equation}
For the first term of the RHS of \eqref{eq:u Gamma}, we have
\begin{equation}
 \left|\intS \Gamma(\gtta, t; \gtta', 0, \go_0)(\gamma(\gtta, \go) - \gamma(\gtta, \go_0))\dd\gtta'\right|\leq
\frac{C}{\sqrt{t}}\intS|\gamma(\gtta, \go) - \gamma(\gtta, \go_0)|\dd\gtta'\, ,
\end{equation}
which converges to $0$, for fixed $t>0$, by hypothesis \eqref{eq:initial condition u}.

Secondly, it is easy to see from the definition \eqref{eq:pt Gamma} of the density $p$ and the estimates \eqref{eq:control Gamma} and \cite[Th.9
p.263]{cf:Friedman} concerning the fundamental solution $\Gamma$ that both $p_t(\gtta, \go)$ and $\partial_\gtta
p_t(\gtta, \go)$ are bounded uniformly on $(t, \gtta, \go)\in [0, T]\times \bbS \times \Supp(\mu)$. In particular, a
standard result shows that for fixed $(t, \gtta)$, the second term of the RHS of \eqref{eq:u Gamma} goes to $0$ as
$\go\to\go_0$. But then the joint continuity of $p$ at $(t_0,\gtta_0, \go_0)$ follows from \eqref{eq:pt Gamma} and uniform estimates on $\Gamma$
(see \cite[Th.9 p.263]{cf:Friedman}).
\end{proof}

\end{appendix}

\end{document}